\documentclass[11pt,a4paper,reqno]{amsart}[tikz,border=10pt]%
\usepackage{amsthm,amsmath,amsfonts,amssymb,amsxtra,appendix,bookmark,dsfont,bm,mathrsfs}
\usepackage[margin=30mm]{geometry}
\usepackage{braket}
\usepackage{caption}  
\usepackage{enumitem}
\usepackage{marginnote}

\theoremstyle{plain}
\newtheorem{theorem}{Theorem}
\newtheorem{lemma}[theorem]{Lemma}

\theoremstyle{definition}

\theoremstyle{remark}
\newtheorem{remark}[theorem]{Remark}

\usepackage{verbatim}

\usepackage{pgfplots}

\usepackage{tikz}   

\DeclareMathOperator{\Span}{Span}
\newcommand{\cl}{\mathrm{cl}}
\newcommand{\exc}{\mathrm{exc}}

\newcommand{\eff}{\mathrm{eff}}
\newcommand{\Bog}{\mathrm{Bog}}

\def\bq{\begin{eqnarray}}
\def\eq{\end{eqnarray}}
\def\bqq{\begin{align*}}
\def\eqq{\end{align*}}
\def\be{\begin{equation}}
\def\ee{\end{equation}}

\def\nn{\nonumber}

\renewcommand{\epsilon}{\varepsilon}

\def\ri{\mathrm{i}}

\def\cF {\mathcal{F}}

\def\cH{\mathcal{H}}
\def\cK{\mathcal{K}}

\def\R {\mathbb{R}}

\def\C {\mathbb{C}}

\def\Z {\mathbb{Z}}

\def\R {\mathbb{R}}
\def\C {\mathbb{C}}

\def\d{\,{\rm d}}

\newcommand{\kb}{{\mathbf k}}
\newcommand{\lb}{{\mathbf l}}

\newcommand{\pb}{{\mathbf p}}
\newcommand{\qb}{{\mathbf q}}

\newcommand{\xb}{{\mathbf x}}

  \newcommand{\e}{\mathrm{e}}

\allowdisplaybreaks

\title[Beliaev damping in Bose gas]{Beliaev damping in Bose gas}

\author[J. Derezi\'nski]{Jan  Derezi\'nski}
\address{Department of Mathematical Methods in Physics, Faculty of Physics, University of Warsaw,  Pasteura 5, 02-093 Warszawa, Poland}
\email{jan.derezinski@fuw.edu.pl} 
\author[B. Li]{Ben Li}
\address{Department of Mathematical Methods in Physics, Faculty of Physics, University of Warsaw,  Pasteura 5, 02-093 Warszawa, Poland}
\email{ben.li@fuw.edu.pl} 
\author[M. Napi\'orkowski]{Marcin Napi\'orkowski}
\address{Department of Mathematical Methods in Physics, Faculty of Physics, University of Warsaw,  Pasteura 5, 02-093 Warszawa, Poland}
\email{marcin.napiorkowski@fuw.edu.pl}

\begin{document}
\date{\today}
\reversemarginpar

\begin{abstract}  According to the Bogoliubov theory the low energy behaviour
  of the Bose gas at zero temperature can be described by non-interacting
  bosonic quasiparticles called phonons. In this work 
  the damping rate of phonons at low momenta, the so-called Beliaev
  damping, is explained and computed with  simple arguments
  involving the Fermi Golden Rule and Bogoliubov's quasiparticles.
\end{abstract}

\maketitle

\section{Introduction}

The  Bose gas near the zero temperature has
curious properties that
can be partly explained from the first principles  by a beautiful argument 
that goes back to Bogoliubov \cite{Bogoliubov-47}. In Bogoliubov's  approach the Bose gas
at zero temperature
can be approximately described by a  gas of weakly interacting
quasiparticles. The  dispersion relation of these quasiparticles, that
is, their energy in function of the momentum is described by a function
 $\kb\mapsto e_\kb$ with an interesting shape. At low
momenta these quasiparticles are called phonons and $e_\kb\approx ck$, where
$c>0$ and $k:=|\kb|$. Thus the low-energy dispersion relation is very
different from the non-interacting, quadratic one. It is responsible
for superfluidity  of the Bose gas.

It is easy to
 see that phonons   could be metastable, because the energy-momentum
 conservation  may not
 prohibit them to decay into two or more phonons. This decay rate was first
 computed in perturbation theory by Beliaev \cite{Beliaev-58}, hence the name {\it Beliaev damping}. According to his computation, the
 imaginary part of the dispersion relation behaves for small momenta as $-
 c_\mathrm{Bel}k^5$.
 This implies the  exponential decay of phonons with the
 decay rate
 $2c_\mathrm{Bel}k^5$. The Beliaev damping has been observed in
 experiments, and appears to be consistent with its theoretical
 predictions
 \cite{Katzetall-02,Hodbyetall-01}.

 In our paper we present a systematic derivation of  Beliaev
 damping. Our presentation differs in several points from similar
 accounts found in the physics literature. We try to make all the arguments
 as transparent as possible, without hiding some of less rigorous
 steps. We avoid using diagrammatic techniques, in favor of
 a mathematically much clearer picture involving a Bogoliubov
 transformation and the 2nd order perturbation computation (the
 so-called Fermi Golden Rule) applied to what we call the effective Friedrichs Hamiltonian.
 We use the grand-canonical picture instead of the canonical one
 found in a part of the  literature. This is a minor difference;
  on this level both pictures should lead to the same final result.
We believe that the derivation of Beliaev damping is a beautiful
illustration of methods
 many-body quantum physics, which is quite convincing even if not fully rigorous.

In the remaining part of the
introduction we provide a brief sketch of the main steps of 
Beliaev's argument. In the main body of our article we discuss these steps
in more detail, indicating which parts can be easily made rigorous.

Let $v$ be a  real function satisfying $v(x)=v(-x)$.
  Later on
  we will need more
assumptions: in particular, we will assume that $v(x)$ is rotationally invariant,   both $v(x)$ and   its
Fourier transform $\hat v(\kb)$   decay
sufficiently fast at infinity and that  $\hat v(\kb)\geq0$.
The homogeneous Bose gas of $N$ particles interacting with the  pair
potential $v$ is described by the Hamiltonian and the total momentum
\begin{align}\label{papo}
  H_N&=-\sum_{i=1}^N\frac{1}{2m}\Delta_i+\sum_{1\leq i<j\leq N}v(x_i-x_j),\\
P_N&=\sum_{i=1}^N\frac{1}{\ri}\partial_{x_i}.\label{papo1}\end{align}
These operators act on $L_\mathrm{s}^2\big((\R^3)^N\big)$, the space
of functions symmetric in the positions of $N$
3-dimensional particles.
Note that $H_N$ commutes with $P_N$, which expresses the 
 spatial homogeneity of the system.  From now on we will set $m=1$.

 We would like to describe a Bose gas of positive density in infinite
volume. This is difficult to do in terms of the Hamiltonian acting on
the whole space $\mathbb{R}^3$. Therefore we replace \eqref{papo}
and \eqref{papo1} with a system enclosed
in a box of size $L$, and then take the thermodynamic limit. In order to
preserve translation symmetry we consider periodic boundary
conditions.
They are not very physical, but it is believed that they should not
affect the overall picture in the thermodynamic limit.

Thus $v$ is replaced by its periodized version adapted to the box
of size $L$. The new
Hilbert space is $L_\mathrm{s}^2\big(([-L/2,L/2]^3)^N\big)$. We will use the same
symbols $H_N,P_N$ to denote the Hamiltonian and total momentum in the
box. Note that they still commute with one another.

It is very convenient to consider at the same time  all numbers of
particles. In order to control the density, that is $\frac{N}{L^3}$, we introduce
the chemical potential given by a positive number $\mu>0$, and we
use the grand-canonical formalism. It is also convenient to pass from
the position to the momentum representation.
Thus we replace $H_N,P_N$ with
\begin{align}
  H&:=\mathop\oplus\limits_{N=0}^\infty (H_N-\mu N)=\int a_x^*\Big(-\frac{1}2\Delta_x-\mu\Big)  a_x\d x+
\frac12  \int\int\d   x\d  y v(x-y)a_x^ *  a_y^ *  a_y a_x\notag\\
    &=\sum_\pb \Big(\frac12\pb^2-\mu\Big)a_\pb^ *  a_\pb\d \pb\label{haha1a}+
\frac{1}{2L^{3}}\sum_\pb\sum_\qb\sum_\kb\hat v(\kb)a_{\pb+\kb}^ *
      a_{\qb-\kb}^ *  a_\qb a_\pb,\\
P&:=\mathop\oplus\limits_{N=0}^\infty P_N=\int a_x^*
   \frac{1}{\ri}\partial_x a_x\d x=\sum_\pb \pb a_\pb^ *  a_\pb.\label{haha2a}
\end{align}
$a_x^*$ and $a_x$ are the
usual creation/annihilation operators for $x\in[-L/2,L/2]^3$ in the
position representation, commuting
to the Dirac delta. $a_\pb^*, a_\pb$ are the usual
creration/annihilation operators for $\pb\in2\pi\mathbb{Z}^3/ L$ in the momentum representation 
commuting to the Kronecker delta.
 $H,P$ act on the bosonic Fock space with the one-particle space
 $L^2\big([-L/2,L/2]^3\big)$ in the position representation, and
 $l^2\big(2\pi\mathbb{Z}^3/ L\big)$ in the momentum
 representation.
 $H$ and $P$ still commute with one another.  

Now there comes the main idea of the Bogoliubov approach. At zero
temperature, one expects complete Bose--Einstein condensation. This is
expressed by assuming
that  the zero mode is populated macroscopically and  there are only very few particles in nonzero
modes.
The zero mode is treated
classically, and essentially removed from the picture. One obtains an approximate
Hamiltonian, which does not preserve the number of particles. One
argues that its most important component is the quadratic part which
involves operators of the form $a_\kb a_{-\kb}$,  $a_\kb^* a_{-\kb}^*$
and $a_\kb^*a_\kb$, $\kb\neq0$. It can be diagonalized by  a linear
transformation which mixes
creation and annihilation operators, called since \cite{Bogoliubov-47} a {\em Bogoliubov
transformation}, and becomes
\begin{eqnarray}\label{hambo}
  \qquad H_\Bog&:=& \sum_{\kb\neq 0} e_\kb b_\kb^* b_\kb,\\
  \label{coeff:bogdisp1}
e_\kb&:=& \sqrt{\frac{1}{4}|\kb|^4+
          \frac{\hat{v}(\kb)}{\hat{v}(0)}\mu|\kb|^2}.
\end{eqnarray}
Thus, the Bogoliubov approximation states that
\begin{equation}\label{eq:Bog_approx}
  H \approx E_\Bog+H_\Bog 
\end{equation}
where $E_\Bog$ is a constant, which will not be relevant for our analysis.
The operator $b_\kb^*$ is the creation operator
of the  {\em quasiparticle} with momentum $\kb$. It is a linear combination of $a_\kb^*, 
a_{-\kb}$. \eqref{hambo} is sometimes called a {\em Bogoliubov
Hamiltonian}. It describes independent quasiparticles with the
{\em dispersion relation} $e_\kb$. The {\em Bogoliubov vacuum}, annihilated by
$b_\kb$ and denoted $\Omega_\Bog$, is its ground state, and can be treated as an approximate
ground state of the many-body system.
The Bogoliubov Hamiltonian
 is still translation invariant: in fact, it commutes with the total momentum,
 described (without any approximation) by
 \be
 P= \sum_{\kb\neq 0} \kb b_\kb^* b_\kb.\label{hambo.}
 \ee

It is easy to describe the thermodynamic limit of \eqref{hambo} and  \eqref{hambo.}: we simply replace the
summation  by integration, without changing the
dispersion relation:
\begin{align}\label{hambo1}
  H_\Bog&= \int e_\kb b_\kb^* b_\kb\d\kb,\\
 P&= \int \kb b_\kb^* b_\kb\d\kb.
  \end{align}

   It is interesting to visualize 
 possible energy-momentum values predicted by the Bogoliubov approximation or, in
        a more precise mathematical language, the joint spectrum of
        the total momentum $P$ and the Bogoliubov Hamiltonian
        $H_\mathrm{Bog}$. 
        On the 1-quasiparticle
        space this joint spectrum is given by the graph of the
        function $\kb\mapsto e_\kb$.
  On fig. \ref{fig1} 
   we show a typical form of the dispersion
relation in the low momentum region, marked with the black line. We
assume that the potential $v$ satisfies the usual assumptions stated
before \eqref{papo} and the second derivative of $\hat v$ near zero is small 
  enough. 
The green line 
denotes the bottom of the 2-quasiparticle spectrum,
that is the joint spectrum of $(H_\Bog,P)$ in the 2-quasiparticle sector.
The bottom of 
the full joint spectrum of
$(H_\Bog,P)$ is marked with an orange dashed line. \footnote{
  Strictly speaking, figures 1 and 2  should be
    interpreted as follows.
    We choose coordinates, so that $P=(P_1,P_2,P_3)$.
    We assume that $P_2=P_3=0$, on the
   horizontal axis we put $P_1$, and on the vertical axis $H$. The full
   $4$-dimensional joint spectrum is rotationally invariant in $P$,
   hence easily reconstructed from our pictures.}  For more details concerning the construction of the excitation spectrum we refer to \cite{CorDerZin-09,DerNap-13,DerMeiNap-13}.
  
\begin{figure}[ht]
  \centering 
  \includegraphics[width=0.7\textwidth]{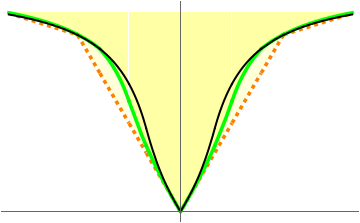}   
 \caption{Joint spectrum of $(H_\Bog,P)$ for generic potentials}
\label{fig1}
\end{figure}

  One can perform an additional step in the Bogoliubov approach. If
  the potential $v$ has a very small support, one can argue that
  $\frac{\hat v(\kb)}{\hat v(0)}$ can be approximated by $1$. One then 
  usually says that the interaction is given by {\em contact potentials},
  which  in the physics literature are often presented in the position representation
  as $v(x)=4\pi a\delta(x)$, where $a$ is a constant, called the scattering
  length.
Strictly speaking, this is however  not correct. The delta function
  needs a renormalization to become a well-defined interaction in the
  two-body case; in the $N$-body case the situation is even more
  problematic.
 In some cases one can  justify this approximation  in
the dilute case using the so-called
{\em Gross-Pitaevski limit}.
Anyway, in this approximation we
  obtain a simpler dispersion relation
\be\label{coeff:bogdisp2} 
e_\kb= \sqrt{\frac{1}{4}|\kb|^4+
          \mu|\kb|^2}.
        \ee
        On fig. \ref{fig2} we show the  energy-momentum
        spectrum corresponding to \eqref{coeff:bogdisp2}.
\begin{figure}[ht]
  \centering 
  \includegraphics[width=0.7\textwidth]{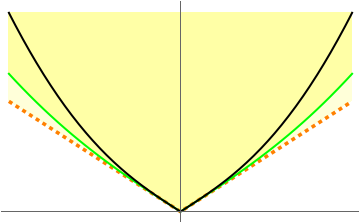}   
 \caption{Joint spectrum of $(H_\Bog,P)$ for contact potentials}
\label{fig2}
\end{figure}

The Hamiltonian $H_\Bog$, both with the dispersion relation 
        \eqref{coeff:bogdisp1} and \eqref{coeff:bogdisp2} has
        remarkable physical consequences.
Note first that the dispersion relation $\kb\mapsto e_\kb$ has a linear cusp at the
        bottom. It also has a positive critical velocity, that is,
        \be
        c_\mathrm{crit}:=\sup\{c\ |\ e_\kb\geq ck,\quad \kb\in\mathrm{R}^3\}>0.\ee
        In other words, the graph $\kb\mapsto e_\kb$ is above
        $\kb\mapsto c_\mathrm{crit} k$. The full joint spectrum
        $\sigma(P,H_\Bog)$ is still above         $\kb\mapsto
        c_\mathrm{crit} k$.  This is interpreted as one of the most
        important properties of superfluidity: a droplet of the Bose
        gas travelling with velocity less than $
        c_\mathrm{crit} k$ has negligible friction (see e.g. \cite{CorDerZin-09}).

        Of course, $H_\Bog$ yields only an approximate description of
        the Bose gas. In reality, one cannot treat the quasiparticles
        given by $b_\kb^*,b_\kb$ as fully independent. In the
        derivation of the Bogoliubov Hamiltonian various terms were
        neglected. In particular, terms of the third and fourth degree
        in $b_\kb^*,b_\kb$ were dropped. Replacing $v$ by $\kappa v$
        we obtain  an (artificial)  coupling constant, to be set to $1$ at the end.
        The third order terms are multiplied by
        $\sqrt{\kappa}$
and the  quartic
        terms by $\kappa$. We argue that the quartic terms are of
        lower order and can  be dropped.      The third order terms have the form
        \begin{align}\label{triple0}
&\frac{1}{\sqrt{L^3}}\sum_{\kb,\pb,\kb+\pb\neq0}
                        u_{\kb,\pb}b_\kb^*b_\pb^*b_{\kb+\pb}+
                        \overline{u_{\kb,\pb}}b_{\kb+\pb}b_\kb^*b_\pb^*\\
          +&\frac{1}{\sqrt{L^3}}\sum_{\kb,\pb,\kb+\pb\neq0}
                        w_{\kb,\pb}b_\kb^*b_\pb^*b_{-\kb-\pb}^*+
                        \overline{w_{\kb,\pb}}b_{-\kb-\pb}b_\kb b_\pb.\label{triple}
                        \end{align}
We will argue (see Section \ref{sec:cubicirrelevant}) that  triple creation and triple annihilation terms
do not contribute to the decay of phonons. Thus we drop
also \eqref{triple}.

Let us investigate what happens with the quasiparticle state
$b_\kb^*\Omega_\Bog$ under the perturbation
\eqref{triple0}.  To this end we first
  need to check with which states have non-zero matrix
  elements with  $b_\kb^*\Omega_\Bog$.  We easily see that  it is directly
  coupled by \eqref{triple0}  only to the
2-quasiparticle sector. By taking the thermodynamic limit we can assume
that the variable $\kb$ is continuous. Thus the perturbed
quasiparticle can be described by the space $\mathbb{C}\oplus L^2(\R^3/\Z_2)$
with the Hamiltonian
\begin{align}
\label{def:themodel0}
 H_{\rm{Fried}}(\kb)&:=\begin{bmatrix} 
 	e_\kb & (h_\kb|\\
 	|h_\kb) & e_\pb+e_{\kb-\pb}
      \end{bmatrix},\end{align}
    where $h_\kb$ can be derived from \eqref{triple0}.
   Here, the  action of $\Z_2$ on $\R^3$ is
  $\pb\mapsto\kb-\pb$, and is related to the Bose symmetry
  $b_\pb^*b_{\kb-\pb}^*\Omega_\Bog= b_{\kb-\pb}^*b_\pb^*\Omega_\Bog$.

Hamiltonians similar to \eqref{def:themodel0}
are well understood. They are often used as
toy models in quantum physics, and are sometimes  called {\em
  Friedrichs Hamiltonians}.

 It is important to notice that  if we set $ h_\kb=0$, so that
  the off-diagonal terms in \eqref{def:themodel0} disappear,
the unperturbed quasiparticle energy
$e_\kb$ lies inside the continuous spectrum of 2-quasiparticle
excitations
$\{e_\pb+e_{\kb-\pb}\ |\ \pb\in\mathbb{R}^3\}$, at least for small momenta. (To be able to
say this we need the
thermodynamic limit which makes the momentum  continuous.) 
To see this, note that if $\kb\mapsto e_\kb$ is convex
we have a particularly simple expression (cf. Lemma \ref{lem:convex_shape})
for the infimum of the 2-quasiparticle spectrum:
 \be \label{def:infconv.}
 \inf_{\pb}\{e_\pb+e_{\kb-\pb}\}= 2 e_{\kb/2}.
 \ee
 Now  \eqref{coeff:bogdisp2}
 is strictly convex, hence $e_\kb$
 lies inside the continuous spectrum of 2-quasiparticle
excitations.  If the second derivative of $\hat v$ near $0$
  is small enough, then the generic dispersion relation 
\eqref{coeff:bogdisp1} is convex for small momenta, hence
then this  property is true at least for small momenta.

Because of that, one can expect that
the
position of the singularity of the resolvent
of \eqref{def:themodel0} 
becomes complex---it describes a resonance and not a bound state. This is interpreted as the unstability of the quasiparticle:
its decay rate is twice the  imaginary part of the 
resonance.

The second order perturbation theory, often called
the {\em Fermi Golden Rule},
says that in order to compute the
(complex) energy shift of an eigenvalue we need  to find the so-called
self-energy $\Sigma_\kb(z)$, which for  $z\not\in\R$ in our case is given by
the integral
    \begin{align}
       \Sigma_\kb(z)&=\frac{1}{2(2\pi)^3}\int\frac{h_\kb^2(\pb)\d\pb}{(z-e_\pb-e_{\kb-\pb})}. \label{droppo4.} 
\end{align}
Then $\Sigma_\kb(e_\kb+\ri0)$ should give the energy shift of
the eigenvalue $e_\kb$.

The imaginary part of this shift is much easier to compute. In fact,
 let $\mathcal{P}\frac1x$ denote the principal value of $\frac1x$.
Applying the Sochocki-Plemelj formula 
\begin{equation}
\frac1{x+\ri0}=\mathcal{P}\frac{1}x-\ri\pi\delta(x) , \label{eq:sochocki}
\end{equation}
we obtain
\begin{align}
  \mathrm{Im}
  \Sigma_\kb(e_\kb+\ri0)&=\frac{-1}{16 \pi^2}\int h_{\kb}^2(\pb)\delta(e_\kb-e_\pb-e_{\kb-\pb})\d\pb. \label{droppo4..} 
\end{align}
%

  In the main result of our paper we make an assumption which
  is a compromise between the usual regular case and a contact
  potential.
  We assume that $\hat v$, the Fourier transform
of the potential in a neighborhood of zero is 
constant, however it decays for large $\kb$ sufficiently fast.
 In Theorem \ref {thm:damping} we prove that
under these assumptions
\be\label{damping1..}
 \mathrm{Im}  \Sigma_{\kb}
  (e_k+\ri0)=-c_\mathrm{Bel}k^5+
O(k^6)  \qquad \text{as} \qquad k\to 0,\qquad c_\mathrm{Bel}=\frac{3\hat{v}(0)}{640\pi \mu}.
  \ee
Physically \eqref{damping1..} means that quasiparticles are almost stable for small $k$ with the lifetime
proportional to $k^{-5}$.


We remark that our analysis is based on the grand-canonical approach
where $\mu$ is the chemical potential.
 In the canonical picture the dispersion relation in the thermodynamic limit is conjectured to be 
\be e_{\kb}=\sqrt{\frac14|\kb|^4 +4\pi a \rho |\kb|^2}.\label{canonical}\ee
Comparing \eqref{coeff:bogdisp2}  with \eqref{canonical} we obtain   $4\pi a \rho\approx
\mu$. Actually, at positive temperatures $\rho$ should be replaced by 
 the condensate density $\rho_0$. It is well-known that for weak
 potentials $\hat v(0)\approx 4\pi a$. Thus \eqref{damping1..} can be
 rewritten as
\be \label{eq:canonical_cbel}
c_\mathrm{Bel}=\frac{3}{640\pi\rho_0},\ee
which is the form of the Beliaev constant usually stated in the
physics literature (\cite{ShiGri-98,Giorgini-98,Liu-97,Chung-09}).  In
particular, to the leading order the damping rate depends on the potential only through $\rho_0$.

The Fermi Golden Rule predicts that  the real part of the dispersion
relation of the interacting system is approximately given by
$e_\kb+\mathrm{Re}\Sigma_\kb(e_\kb+\ri0)$, where
    \begin{align}
       \mathrm{Re}\Sigma_\kb(e_\kb+\ri0)&=\frac{1}{2(2\pi)^3}\mathrm{Re}\int\frac{h_{\bf k}(\pb)^2 \d\pb}{(e_\kb-e_\pb-e_{\kb-\pb}+\ri0)}. \label{droppo4-} 
\end{align}
If $\hat v(\kb)$ has a sufficient decay and we use the
dispersion relation \eqref{coeff:bogdisp1}, then
\eqref{droppo4-}  is well defined. However if we use the formula
\eqref{coeff:bogdisp2}  for
contact potentials, then \eqref{droppo4-}  is divergent for large
$\kb$. This is related to the fact that constant $\hat v(\kb)$ does
not correspond to a well-defined potential (one  has to renormalize
its coupling constant).

Unfortunately,
\eqref{droppo4-}  yields an unphysical prediction for
small momenta. Under the same assumptions as in the
main theorem, we show that
$\lim\limits_{\kb\to0}\Sigma_\kb(0)=\infty$. Thus the Fermi Golden Rule
predicts an infinite energy shift at zero momenta, which is certainly incorrect. This is in agreement with second-order perturbation theory results from physics literature \cite{ShiGri-98}.


 Similar results about both imaginary and real part of the
  shift of the dispersion relation can be obtained for more
general potentials. We indicate possible generalizations of our result
in remarks.

One can conclude that  perturbation theory around the Bogoliubov Hamiltonian
 provides  a reasonable method to find the second order imaginary
 correction to the dispersion relation.  However, the  computation of
 its real part seems more dubious, at least
 for small momenta.

The above problem is an indication of the crudeness of the
Bogoliubov approximation.  Throwing out the zero mode from the picture
(or, which is essentially the same, treating it as a classical
quantity), as well as throwing out higher order terms, is a very
violent act and we should not be surprised by a punishment. By the
way, one expects that the true dispersion relation of phonons goes to
zero as $\kb\to0$. This is the content of the so called
``Hugenholtz-Pines Theorem'' \cite{HugPin-59}, which is  a
(non-rigorous) argument  based on the gauge invariance. Perturbation
theory around the Bogoliubov Hamiltonian is compatible with this
theorem where it comes to the imaginary part. For the real part it
fails.

Better results of computations of the imaginary part over the real
part
based on the Fermi Golden Rule are not very surprising. 
It is a general property of
 Friedrichs Hamiltonians with singular off-diagonal terms: the
 imaginary part of the perturbed eigenvalue can be computed much more
 reliably than its real part. We describe this phenomenon briefly in Sections
 \ref{Friedrichs    Hamiltonian} and \ref{Fermi Golden Rule}.

Readers who like clean mathematical results illustrating physical
phenomena (which includes the authors)
may be somewhat dissatisfied with a relatively long chain of
arguments  presented in this paper. One of its aspects is the use of a finite
system in a box at some of the steps (e.g. Bogoliubov approximation
and removal of the zeroth mode), and of the thermodynamical limit in
others (computation of the resonance, which requires continuous
spectrum, hence, infinite volume). Unfortunately, we do not know 
a better description. We are just trying to follow the
usual physicist's reasoning, without hiding its non-rigorous steps.

Let us now make a few remarks about the literature. The theory of metastable states and their exponential decay goes back
to the work of Dirac \cite{dirac}.
The concept of a resonance as a pole on the ``unphysical sheet of the
complex plane'' is usually attributed to Wigner-Weisskopf
\cite{wigner-weisskopf}. It is discussed, including its historical background,
in Chap. XII.6 of \cite{reed-simon-4}.
In his  lecture notes
\cite{fermi} Fermi
formulated
 two ``golden rules'' that describe 2nd order theory for
 eigenvalues and their decay rate.
 The
Friedrichs Hamiltonian is a useful pedagogical toy model, which nicely  ilustrates
the Fermi Golden Rule. It  goes back to \cite{Friedrichs-65}, see also
\cite{DerFru-02}.
An elegant
rigorous description of exponential decay expressing the Fermi Golden
Rule
was given by Davies in his theorem about the weak coupling limit
\cite{Davies-74}, see also \cite{DerdeR-07,DerdeR-08}.

The original paper
of Bogoliubov \cite{Bogoliubov-47} was heuristic, however in recent years there
have been many rigorous papers justifying Bogoliubov's approximation
in several cases.  The first result justifying  \eqref{eq:Bog_approx} has been obtained in the mean-field scaling by Seiringer in \cite{Seiringer-11} (see also \cite{LewNamSerSol-15,DerNap-13,GreSei-13,NamSei-15} for related results). Recently, corresponding results have been obtained in the Gross-Pitaevskii regime \cite{BocBreCenSch-19,BreSchSch-22,NamTri-23} and even beyond \cite{BreCapSch-22}.  A time-dependent version of Bogoliubov theory has been successful in describing the dynamics of Bose-Einstein condensates and excitations thereof (see \cite{NamNap-17b,Napiorkowski-23} for reviews).     

As explained above, to describe damping one has to go beyond
Bogoliubov theory.
In the mean-field regime this has been done for the ground state
energy expansion in \cite{NamNap-21,BosPetSei-21},
 including singular interactions \cite{BosLeoPetRad-23}, and for the dynamics
in \cite{BosPetPicSof-22}. Very recently, the results beyond Bogoliubov theory have been obtained in the Gross-Pitaevskii regime
\cite{CarOlgSASch-23}.

 None of the above rigorous papers, with exception of \cite{DerNap-13},
addressed the energy-momentum spectrum. In fact,  it is very difficult to study rigorously the dispersion relation in the 
thermodynamic limit---which is essentially  necessary to analyze
phonon damping.

The quasiparticle  picture of the  Bose gas at low temperatures  has been confirmed
in experiments. The dispersion relation of  ${}^4\mathrm{He}$ can
be observed in neutron scattering experiments, and is remarkably
sharp. It has
been measured within a large range of wave numbers covering not only
phonons, but also the so-called maxons and rotons, see e.g.
\cite{Godfrinetall-21}.
In particular, one can see that the dispersion relation is slightly
higher than the 2-quasiparticle spectrum for low wave numbers.
The quasiparticle picture has  also been confirmed by
experiments on Bose Einstein condensates involving alkali atoms.
The Beliaev damping has been  observed in experiments on Bose Einstein
condensates. The results are consistent with theoretical
 predictions  \cite{Katzetall-02,Hodbyetall-01}. Note, however, that
 the precise prediction \eqref{droppo4..} is difficult to verify
 experimentally. Bose-Einstein condensates created in labs are not
 very large, so it is difficult to probe the large wavelength region.

 Let us mention that there exists another  related phenomenon found in
 Bose-Einstein condensates, the so-called  Landau damping, which
 involves instability of quasiparticles due to thermal
 excitations. The Landau damping is absent at zero temperature and
 becomes dominant at higher temperatures. 
The Beliaev damping occurs at zero temperature, and for very small
temperatures it is still stronger than the Landau damping.

In the physics literature, the damping of phonons was first computed
by Beliaev \cite{Beliaev-58}. Landau damping has been for the first time computed by Hohenberg and Martin in \cite{HohMar-65} (see also \cite{MohMor-60}).  Both these results have been reproduced in \cite{ShiGri-98}, also using the formalism of Feynman diagrams and many-body Green's functions. In \cite{Liu-97} the damping rate was derived starting from an effective action in the spirit of Popov’s hydrodynamical approach. \cite{Giorgini-98} repeated the same computation in the time-dependent mean-field approach. In \cite{Chung-09} the mean-field and hydrodynamic approaches were  applied to the 2D case. Our derivation is consistent with the above works, however, in our opinion, avoids some unnecessary elements obscuring the simple mechanism of the Beliaev damping.

The plan of the paper is as follows. Sections \ref{Friedrichs 
  Hamiltonian} and \ref{Fermi Golden Rule} concern general well-known
facts about   about 2nd order perturbation theory of
  embedded eigenvalues. In Section  \ref{sec:derivation} we define the Bose gas Hamiltonian and describe
the Bogoliubov approach in the grand-canonical setting.
 In Section \ref{Effective Friedrichs Hamiltonian} we derive heuristically the
 effective model that we consider. Then, in Section
 \ref{sec:cubicirrelevant} we discuss the shape of the energy-momentum
 spectrum and explain why the contribution from term
 \eqref{triple} is irrelevant for the damping rate computation, which is
 the main result of the paper is 
proven  in Section \ref{sec:dampingrate} as
Theorem \ref{thm:damping}. The analysis of
 the real part of the
 self-energy, and of its (unphysical) behavior at small momenta by the method of this paper is described in Section
 \ref{sec:renorm}.

\section{Friedrichs   Hamiltonian}\label{Friedrichs  Hamiltonian}
Suppose that $\cH$ is a Hilbert space with a self-adjoint operator
$H$. Let $\Psi\in\cH$ be a normalized vector. We can write
$\cH\simeq\C\oplus\cK$, where $\C\simeq\C\Psi$ and $\cK:=\{\Psi\}^\perp$.
 First assume
that $\Psi$ belongs
to the domain of $H$ and set 
\be 
E_0:=(\Psi|H\Psi),\quad h:= H\Psi-E_0\Psi
.\label{braket}\ee
 Note that $h\in\cK$.
  Let $K$ denote $H$ compressed to
$\cK$.  That means, if $I:\cK\to\cH$ is the embedding, then $K:=I^*H I$.
Then in terms of $\C\oplus\cK$ we can write
\be
H=\begin{bmatrix}E_0&(h|\\|h)&K\end{bmatrix}.
\label{fried}\ee

Operators of this form were studied by Friedrichs in \cite{Friedrichs-65}. Therefore,
sometimes they
 are referred to as {\em Friedrichs Hamiltonians},  e.g. in \cite{DerFru-02,DerdeR-07}:

Let $z\in\C$. The following identity is a special case of the so-called {\em
  Feshbach-Schur formula}:
\begin{align}\label{feshbach}
(\Psi|(H-z)^{-1}\Psi)&=\frac{1}{E_0+\Sigma(z)
                                -z},\\
  \Sigma(z)&:=-(h|(K-z)^{-1}h).\label{feshbach1a}\end{align}
Following a part of the physics literature, we will call $\Sigma(z)$
the {\em self-energy}.
For further reference let us rewrite \eqref{feshbach}
as
\be\label{fesh1}
\Sigma(z)=\frac1{(\Psi|(H-z)^{-1}\Psi)}+z-E_0.\ee
 Note that  the full resolvent of $H$ can be computed, see
  e.g. \cite{DerFru-02}, or  Equation (1.2) of \cite{DeGe}:
\begin{align} \label{eq:FriedRes}
(H-z)^{-1}=&\begin{bmatrix}0&0\\0&(K-z)^{-1}\end{bmatrix}\\+&
\begin{bmatrix}1\\(K-z)^{-1}|h)\end{bmatrix}
\frac1{E_0+\Sigma(z)
                                -z}\begin{bmatrix}1&(
                                  h|(K-z)^{-1}
                                \end{bmatrix}.\notag\end{align}

 If $K$ has continuous spectrum, it often happens that $\Sigma(z)$
 can be continued analytically from 
the upper complex halfplane across the spectrum to the {\em non-physical sheet of the 
  complex plane}. Then $(\Psi|(H-z)^{-1}\Psi)$ may have a singularity 
for $z=E=E_\mathrm{R}-\ri \frac\Gamma2$ with $\Gamma>0$. This
singularity $E$ is called a {\em resonance}. Suppose that
$\Gamma$ is small. A well-known 
non-rigorous argument, involving  a change of the contour of 
integration and described e.g. in Chap. XII.6 of \cite{reed-simon-4}  (see also \cite{FetWal-03}) , shows that over a long period of time (not too small and 
not too large) we have 
\be\label{decay}
(\Psi|\e^{-\ri t H}\Psi)\simeq C\e^{-\rm{i} E_\mathrm{R} t-\frac{\Gamma}2t}.\ee 
This is interpreted as  exponential decay of the state $\Psi$ with 
the decay rate $\Gamma$.

We can apply the formulas \eqref{feshbach}- \eqref{eq:FriedRes}
also if $\Psi$ does not belong to the
domain of $H$, but belongs to its form
domain, so that $( \Psi|H\Psi)$ is well
defined. Note that $E_0$ and 
$\Sigma(z)$  are then uniquely defined by 
\eqref{braket} and \eqref{fesh1}.

If $\Psi$ does not belong to the form domain of $H$, then strictly
speaking the self-energy is ill defined. In practice in such
situations one often introduces  a cutoff Hamiltonian $H^\Lambda$, which in
some sense approximates $H$.
Then,
setting
$h^\Lambda:= H^\Lambda\Psi$,
$E_0^\Lambda:=(\Psi|H^\Lambda\Psi)$, and denoting by
$K^{\Lambda}$ the operator $H^\Lambda$ compressed to
$\cK$,  one can use the cutoff version of the Feshbach-Schur formula:
\begin{align}\label{feshbach-r}
(\Psi|(H^\Lambda-z)^{-1}\Psi)&=\frac{1}{E_0^\Lambda+\Sigma^\Lambda(z)
                                -z},\\
  \Sigma^\Lambda(z)&=-(h^\Lambda|(K^{\Lambda}-z)^{-1}h^\Lambda).\end{align}
The resolvent of the original Hamiltonian $H$ can be retrieved \cite{DerFru-02} 
in the limit $\Lambda\to\infty$:
\be
  (H-z)^{-1}=\lim_{\Lambda\to\infty}  (H^\Lambda-z)^{-1}.\ee
Note that $E_0^\Lambda$ is a  sequence of real numbers, typically
converging to $\infty$.  They  can be treated as
{\em counterterms} renormalizing the self-energy $\Sigma^\Lambda(z)$.

\section{Fermi Golden Rule}
\label{Fermi Golden Rule}

The meaning of the self-energy is especially clear
in perturbation theory. Again, let $\Psi$ be a normalized vector in
$\cH$. Consider a family of  self-adjoint operators $H_\lambda=H_0+\lambda V$ such that
$H_0\Psi=E_0\Psi$. In order to avoid discussing 1st order
  perturbation theory we assume that  $( \Psi|V\Psi)=0$.
Let  $h:=V\Psi-\frac{1}{\lambda}E_0\Psi$ and $K_\lambda$ be
$H_\lambda$ compressed to $\cK$. 
Thus we  rewrite \eqref{fried} as
\be
H_\lambda=\begin{bmatrix}E_0&\lambda(h|\\\lambda|h)&K_\lambda \end{bmatrix}.
\label{fried1}\ee
 We extract $\lambda^2$ from the definition of the self-energy,
  so that \eqref{feshbach} and \eqref{feshbach1a} are rewritten
as
\begin{align}\label{feshbach1}
(\Psi|(H_\lambda-z)^{-1}\Psi)&=\big(E_0+\lambda^2\Sigma_\lambda(z)
                                        -z\big)^{-1},\\
  \Sigma_\lambda(z)&:=-(h|(K_\lambda-z)^{-1}h)
                     =\Sigma_0(z)+O(\lambda)
                     .\end{align}
                 Now                   \eqref{feshbach1} has a pole
                                      at
                                      \be
                                      E_0+\lambda^2\Sigma_0(E_0+
                                      \ri0)+O(\lambda^3).\label{fgr}\ee
                                      This is often formulated as the {\em Fermi
                                        Golden Rule}:
                                   the pole of the resolvent,
                                   originally at     an eigenvalue $E_0$, is shifted in
the                                      second order by $\lambda^2\Sigma_0(E_0+
                                      \ri0)$.
                                      This
                                      shift  can have a negative
                                      imaginary part, and then
                                      the eigenvalue disappears, and
                                      instead we have a resonance.

                                      For small couplings $\lambda$ a rigorous meaning of
the decay property \eqref{decay}
 is provided by the               following version of the {\em weak coupling
                                      limit} (\cite{Davies-74}, see also \cite{DerdeR-07,DerdeR-08})
                                  \be
                                      \lim_{\lambda\to0}
                                      \big(\Psi\big|
                                      \exp\big(-\ri\tfrac{t}{\lambda^2}(H_\lambda-E_0)\big)\Psi\big)=
                                      \e^{-\ri t\Sigma_0(E_0+i0)}.
                                      \ee

If the perturbation is
                                      singular, so that $\Psi$ does
                                      not belong to the domain of
                                      $V$,
then $\Sigma_0(z)$ is in general ill defined and \eqref{fgr}  may
  lose its
meaning. Strictly speaking, one then needs to introduce a cutoff on
the perturbation and
a counterterm, and only then to apply the appropriately
modified Fermi Golden Rule.

Note that it is enough to consider real counterterms.
Therefore, if we know that the renormalized energy is 
close to $E_0$, then we can still expect that \eqref{fgr}  gives a
correct prediction for the imaginary part of the resonance.
In other words, the
imaginary part of the
singularity of the resolvent $(H_\lambda-z)^{-1}$ is
\be \lambda^2\mathrm{Im}\Sigma_0(E_0+\ri0)+O(\lambda^3),\ee
where we do not need to cut off the perturbation.

 In practice, we start from a singular expression of the form
\eqref{fried1}. To make it well-defined we need to choose a cutoff
and counterterms. These choices will not affect the imaginary part of
the resonance, however in principle, one can add an arbitrary real
constant to a counterterm, which will affect the real part of the resonance.
Therefore, for singular perturbations it may be more difficult to predict the real part of the resonance.

\section{Bose gas and Bogoliubov ansatz}
\label{sec:derivation}

We consider a homogeneous Bose gas of \(N\) particles with a two-body
potential  described by a function \(  v: \R^3\to \R\) with the Fourier transform  \( \hat{v}(  \kb)=\int_{\R^3} v(x)
\e^{-\ri\kb \cdot \xb } \d\xb \). 
In the grand canonical setting and the momentum representation
such a system is governed by the
(second quantized) Hamiltonian
\be  \label{hamiltonian:grandcanonical-}
H= \int \left( \frac{\kb^2}{2}-\mu \right) a_\kb^*a_\kb\d \kb
+\frac{\kappa}{2(2\pi)^3}\int\d\pb\int\d\qb\int\d\kb\hat{v}(\kb)a_{\pb-\kb}^*a_{\qb+\kb}^*a_{\pb}a_{\qb},
\ee
where \(\mu\ge 0\) is the chemical potential and \(
a_\kb^*/a_\kb\) the creation/annihilation operators for particles of
mode \(\kb\). It acts on the bosonic Fock
space \(\cF= \Gamma_\mathrm{s}\big(L^2(\R^3)\big)\),
and for each $N$ it leaves invariant its $N$-particle sector $L^2_\mathrm{s}\big((\R^3)^N\big)$.
Recall that the creation and annihilation operators satisfy the canonical commutation relation (CCR): 
\be \label{eqn:CCR}
[a_\pb,a_\qb]=0 = [a_\pb^*,a_\qb^*], \ \ [a_\pb,a^*_\qb]=\delta(\pb-\qb), 
\ee 
where \( [\ ,\ ]\) is the usual commutator.
 We introduce the coupling constant $\kappa>0$ mostly for
  bookkeeping purposes; note that in the introduction we set $\kappa=1$.

    In most of the paper we will make the following assumption on the potentials:\\
\begin{subequations}      
\begin{align}v\in L^1(\R^3),&\quad\text{ so that $\hat v$ is a
                             continuous function;}\\
  \hat v(\kb)&\geq0,\quad \kb\in\R^3;\\
     \hat v(\kb)=\hat v(0)>0,&\quad\text{ for $|\kb|<\Lambda$,\quad $\Lambda>0$;}\label{assu-flat}\\
     |\hat v(\kb)|\leq C&(1+|\kb|)^{-\frac12-\epsilon},\qquad \text{for some $\epsilon>0$;} 
\label{assu2}
     \\
v&\text{ is rotationally invariant}&.\end{align}
\label{assum}
\end{subequations}

\begin{remark}
One can relax the condition \eqref{assu-flat} to allow for generic
potentials. One could also consider potentials which for some constant
$\nu$ satisfy
\be\label{assu-generic}
\hat v(\kb)=\hat
v(0)+\frac{\nu}{2} |\kb|^2+O(|\kb|^{2+\epsilon}),\qquad\epsilon>0.\ee
We will comment about possible extensions of our results to potentials
satisfying
\eqref{assu-generic} instead of \eqref{assu-flat}.
  \end{remark}

For the reasons explained in the introduction, we replace the infinite
space $\R^3$ by the torus $[-L/2,L/2]^3$ with periodic boundary
conditions. In the momentum representation the Hamiltonian becomes
\be  \label{hamiltonian:grandcanonical}
H= \sum_{\kb\in 2\pi \Z^3/L} \left( \frac{\kb^2}{2}-\mu \right) a_\kb^*a_\kb+\frac{\kappa}{2L^3}\sum_{\pb,\qb,\kb \in 2\pi \Z^3/L}\hat{v}(\kb)a_{\pb-\kb}^*a_{\qb+\kb}^*a_{\pb}a_{\qb}.
\ee
Note that $\hat v$ is the same function as in 
\eqref{hamiltonian:grandcanonical-}, however it is now sampled only on
the lattice
$2\pi \Z^3/L$.
The commutation relations involve now the Kronecker delta:
\be \label{eqn:CCR_disc}
[a_\pb,a_\qb]=0= [a_\pb^*,a_\qb^*], \ \ [a_\pb,a^*_\qb]=\delta_{\pb,\qb}. 
\ee

Let us now pass to  the quasiparticle representation. To
this end we follow the well-known grand-canonical version of the
Bogoliubov approach (see e.g. \cite{CorDerZin-09}). It involves two unitary transformations. 

The first one is a Weyl transformation that introduces a
macroscopic occupation of the zero-momentum mode, the Bose-Einstein
condensate. (In the canonical version  Bogoliubov approach this
corresponds to
the c-number substitution \cite{LieSeiYng-05}.) To this end, for \(\alpha\in \C\), we introduce the Weyl operator of the mode \(\kb=0\)
\be \label{def:weylop}
W_\alpha=\exp(-\alpha a_0^*+\bar{\alpha}a_0).
\ee
Then
\be \nn
W_\alpha^* a^*_\kb W_\alpha = a^*_\kb-\bar{\alpha} \delta_{\kb,0}=: \tilde{a}^*_{\kb}.
\ee
The new annihilation operators with tildes kill the “new vacuum”
$\Omega_\alpha =W^*_\alpha \Omega$. We express our
Hamiltonian in terms of $\tilde a_\kb^*,\tilde a_\kb$. To simplify the notation, in what follows we drop the tildes and we obtain
\begin{eqnarray*}
 H  &=& -\mu |\alpha|^2+ \frac{  \kappa\hat{v}(0)}{2L^3}|\alpha|^4 +\left(\frac{\kappa\hat{v}(0)}{L^3}|\alpha|^2-\mu  \right)(\alpha a_0^*+ \bar{\alpha} a_0) \\
&+& \sum_{\kb} \left( \frac{\kb^2}{2}-\mu +   \frac{\kappa (\hat{v}(\kb)+\hat{v}(0))}{L^3}|\alpha|^2 \right) a_\kb^* a_\kb + \sum_{\kb }  \frac{\kappa\hat{v}(\kb)}{2L^3} \left(\alpha^2  a_\kb^*a_{-\kb}^*+ 
\bar{\alpha}^2  a_\kb a_{-\kb} \right)\\
&+& \frac{\kappa}{L^3}\sum_{\kb_1, \kb_2}\hat{v}(\kb_1)\left( \bar{\alpha} a_{\kb_1+\kb_2}^*a_{\kb_1}a_{\kb_2}+\alpha  a_{\kb_1}^*a_{\kb_2}^*a_{\kb_1+\kb_2}\right) \\
&+&  \frac{\kappa}{2L^3}\sum_{\kb_1 ,\kb_2 ,\kb_3, \kb_4  }\delta(\kb_1+\kb_2-\kb_3-\kb_4)\hat{v}(\kb_2-\kb_3)a_{\kb_1}^*a_{\kb_2}^* a_{\kb_3}a_{\kb_4}.
\end{eqnarray*}
Note that we have
\begin{equation*}
    ( \Omega_\alpha|H\Omega_\alpha)=-\mu |\alpha|^2+ \frac{ \kappa \hat{v}(0)}{2L^3}|\alpha|^4,
\end{equation*}
and we choose $\alpha =\sqrt{\frac{\mu L^3}{\kappa\hat{v}(0)}}$, so that
$\Omega_\alpha$  minimizes this expectation value. This leads to 
\begin{align} \label{eq:WHWalpha}
  H  &=\kappa^{-1}H_0+H_2+\sqrt\kappa H_3+\kappa H_4,\\\notag
H_0&:=
-\frac{\mu^2L^3}{2 \hat{v}(0)},\\\notag
H_2&:= \sum_{\kb} \left( \frac{\kb^2}{2}  +   \frac{\mu \hat{v}(\kb)}{\hat{v}(0)} \right) a_\kb^* a_\kb + \sum_{\kb }  \frac{\mu \hat{v}(\kb)}{2\hat{v}(0)} \left(   a_\kb^*a_{-\kb}^*+ 
   a_\kb a_{-\kb} \right),\\\notag
  H_3&:= \frac{1}{L^{3/2}
       }\sum_{\kb_1,
  \kb_2}\frac{\hat{v}(\kb_1)\sqrt{\mu}}{\sqrt{\hat{v}(0)}} \left(
  a_{\kb_1+\kb_2}^*a_{\kb_1}a_{\kb_2}+
  a_{\kb_1}^*a_{\kb_2}^*a_{\kb_1+\kb_2}\right), \\\notag
H_4&:=\frac{1 }{2L^3}\sum_{\kb_1 ,\kb_2 ,\kb_3, \kb_4  }\delta(\kb_1+\kb_2-\kb_3-\kb_4)\hat{v}(\kb_2-\kb_3)a_{\kb_1}^*a_{\kb_2}^* a_{\kb_3}a_{\kb_4}. \notag
\end{align}

We extract from the above Hamiltonian all terms containing only
non-zero modes:
\begin{align}\notag
  H_2&=\frac{\mu}{2}(a_0^{*2}+a_0^2+2a_0^*a_0)+H_2^{\rm exc},\\
  H_2^{\rm exc}&:=\sum_{\kb\neq0} \left( \frac{\kb^2}{2}  +   \frac{\mu \hat{v}(\kb)}{\hat{v}(0)} \right) a_\kb^* a_\kb + \sum_{\kb\neq0 }  \frac{\mu \hat{v}(\kb)}{2\hat{v}(0)} \left(   a_\kb^*a_{-\kb}^*+ 
                 a_\kb a_{-\kb} \right);\\\notag
  H_3&=\frac{1}{L^{3/2}}\sum_{\kb}\sqrt{\mu\hat
       v(0)}(a_0^*a_\kb^*a_\kb+a_\kb^*a_\kb a_0)\\\notag&+\frac{1}{L^{3/2}}\sum_{\kb\neq0}\frac{\sqrt{\mu}\hat
       v(\kb)}{\sqrt{\hat
       v(0)}}\big((a_0^*+a_0)a_\kb^*a_\kb+a_0a_\kb^*a_{-\kb}^*+a_0^*a_\kb
       a_{-\kb}\big)+H_3^{\rm exc},
\\\label{exc3}       H_3^{\rm exc}&:= \frac{1}{L^{3/2}}\sum_{\kb_1,
  \kb_2,\kb_1+\kb_2\neq0}\frac{\hat{v}(\kb_1)\sqrt{\mu}}{\sqrt{\hat{v}(0)}} \left(
  a_{\kb_1+\kb_2}^*a_{\kb_1}a_{\kb_2}+
                              a_{\kb_1}^*a_{\kb_2}^*a_{\kb_1+\kb_2}\right) ;\\\notag
  H_4&=\frac{1}{2L^3}\hat v(0)\Big(a_0^*a_0^*a_0a_0+2\sum_{\kb\neq0}a_0^*a_0a_\kb^*a_\kb\Big)\\\notag
&+       \frac{1}{2L^3}\sum_{\kb\neq0}\hat v(\kb)(a_0^*a_0^*a_\kb
       a_{-\kb}+a_0a_0a_\kb^*a_{-\kb}^*+2a_0^*a_ 0a_\kb^*a_\kb)\\\notag&+
                                              \frac{1}{L^3}\sum_{\kb_1,\kb_2,\kb_1+\kb_2\neq0}\hat
                                              v(k_1)\big(a_0^*a_{\kb_1+\kb_2}^*a_{\kb_1}a_{\kb_2}+
                                              a_0a_{\kb_1}^*a_{\kb_2}^*a_{\kb_1+\kb_2}\big)+H_4^{\rm
                                                                         exc},\\\label{exc4}
  H_4^{\rm exc}&:=\frac{1 }{2L^3}\sum_{\kb_1 ,\kb_2 ,\kb_3, \kb_4 \neq0 }\delta(\kb_1+\kb_2-\kb_3-\kb_4)\hat{v}(\kb_2-\kb_3)a_{\kb_1}^*a_{\kb_2}^* a_{\kb_3}a_{\kb_4}. 
\end{align}

Let
  \begin{align}\nn
      \sigma_\kb&=\frac{\sqrt{\sqrt{e_\kb^2+B_\kb^2}+e_\kb}}{\sqrt{2e_\kb}},\quad 
    \gamma_\kb=\frac{\sqrt{\sqrt{e_\kb^2+B_\kb^2}-e_\kb}}{\sqrt{2e_\kb}},
    \\\nn
\beta_\kb&=\cosh^{-1}(\sigma_{\kb}) =\sinh^{-1}(\gamma_{\kb}),
    \\
              e_\kb&                   := \sqrt{\frac{1}{4}|\kb|^4+
       B_\kb|\kb|^2}
                     ,\ \ \  B_\kb:= \frac{\hat{v}(\kb)}{\hat{v}(0)}\mu.\end{align}
                   Sometimes we will write $e_k$, $\sigma_k$,
                   $\gamma_k$, instead of
                    $e_\kb$, $\sigma_\kb$,
                   $\gamma_\kb$.
We are going to apply a Bogoliubov transformation 
\be \label{def: bogtranf}
U_\Bog:=\exp \Bigg(\sum_{\kb \neq 0} \beta_\kb(a_\kb^*a_{-\kb}^*-a_\kb a_{-\kb})\Bigg),
\ee
which transforms non-zero mode operators $a_\kb^*,a_\kb$ into
 quasi-particle operators \( b_\kb^*, b_\kb\): \begin{eqnarray}\nn
b_\kb &:=& U_\Bog a_\kb U_\Bog^* = \sigma_\kb a_\kb+ \gamma_\kb a_{-\kb}^*,\\
b_\kb^* &:=& U_\Bog a_\kb^* U_\Bog^* = \sigma_\kb a_\kb^*+ \gamma_\kb a_{-\kb},
\label{bogg}\end{eqnarray}
Let us also note the relation inverse to \eqref{bogg}:
\begin{eqnarray*}
a_\kb &=& \sigma_\kb b_\kb- \gamma_\kb b_{-\kb}^*,\\
a_\kb^* &=& \sigma_\kb b_\kb^*- \gamma_\kb b_{-\kb}. 
\end{eqnarray*}
It is well known that \eqref{bogg} diagonalizes $H_2^{\rm exc}$ in terms of the quasi-particle operators:
\begin{eqnarray} \label{def:excham}
H_2^{{\rm exc}} &=&  E_{\rm Bog}+H_\Bog,
\end{eqnarray} 
where 
   \begin{align}
  E_{\rm Bog}&:= -\frac{1}{2}\sum_{\kb\neq 0}\left(\frac{1}{2}|\kb|^2+\frac{\hat{v}(\kb)}{\hat{v}(0)}\mu-e_\kb \right) ,\\
  \qquad H_\Bog&:= \sum_{\kb\neq 0} e_\kb b_\kb^* b_\kb. 
  \label{coeff:bogdisp}
\end{align}

We also express  $H_3^{\rm exc}$ in terms of quasiparticles:
\begin{align}
H_3^{\rm exc} &=\frac{1}{L^{3/2} }\sum_{\kb_1,\kb_2,\kb_1+\kb_2\neq0} \frac{\sqrt{\mu}\hat
                v(\kb_1)}{\sqrt{\hat v(0)}}\\
              &\Big(\big(\sigma_{\kb_1+\kb_2}b_{\kb_1+\kb_2}^*-\gamma_{-\kb_1-\kb_2}b_{-\kb_1-\kb_2}\big) 
                \big(\sigma_{\kb_1}b_{\kb_1}-\gamma_{-\kb_1}b_{-\kb_1}^*\big) 
                \big(\sigma_{\kb_2}b_{\kb_2}-\gamma_{-\kb_2}b_{-\kb_2}^*\big)\nn\\
+&                \big(\sigma_{\kb_1}b_{\kb_1}^*-\gamma_{-\kb_1}b_{-\kb_1}\big) 
  \big(\sigma_{\kb_2}b_{\kb_2}^*-\gamma_{-\kb_2}b_{-\kb_2}\big)
  \big(\sigma_{\kb_1+\kb_2}b_{\kb_1+\kb_2}-\gamma_{-\kb_1-\kb_2}b_{-\kb_1-\kb_2}^*\big) \Big).
\nn               \end{align}
After opening the brackets and using $\sigma_\kb=\sigma_{-\kb}$  and $\gamma_\kb=\gamma_{-\kb}$, we
transform this into  
  \begin{align}
  H_3^{\rm exc}&=H_{3,1}^{\rm exc}+H_{3,2}^{\rm exc}, \label{H3exc}\\
  H_{3,1}^{\rm exc}
    =&
    \sum_{\kb_1,\kb_2,\kb_1+\kb_2\neq0}(b_{\kb_1+\kb_2}^*b_{\kb_1}b_{\kb_2}+
  b_{\kb_1}^*b_{\kb_2}^*  b_{\kb_1+\kb_2})\\
  \Bigg(  \frac{\sqrt{\mu}\hat
  v(\kb_1)}{L^{3/2}\sqrt{\hat v(0)}}
   &\big(\sigma_{\kb_1+\kb_2}\sigma_{\kb_1}\sigma_{\kb_2}-\gamma_{\kb_1+\kb_2}\gamma_{\kb_1}\gamma_{\kb_2}+ \gamma_{\kb_1+\kb_2}\sigma_{\kb_1}\gamma_{\kb_2}-\sigma_{\kb_1+\kb_2}\gamma_{\kb_1}\sigma_{\kb_2} \big)\nn \\
+ & \frac{\sqrt{\mu}\hat
  v(\kb_1+\kb_2)}{L^{3/2}\sqrt{\hat v(0)}} \big(
\gamma_{\kb_1+\kb_2}-\sigma_{\kb_1+\kb_2}\big)\gamma_{\kb_1}\sigma_{\kb_2}\Bigg),\nn\\
    H_{3,2}^{\rm exc}          =&
    \sum_{\kb_1,\kb_2,\kb_1+\kb_2\neq0}  
 (b_{-\kb_1-\kb_2}^*b_{\kb_1}^*b_{\kb_2}^*+
    b_{-\kb_1-\kb_2}b_{\kb_1}b_{\kb_2})\\ &
   \frac{\sqrt{\mu}\hat
  v(\kb_1)}{L^{3/2}\sqrt{\hat v(0)}}  \Big(\gamma_{\kb_1}\gamma_{\kb_2}\sigma_{\kb_1+\kb_2} -
     \sigma_{\kb_1}\sigma_{\kb_2}\gamma_{\kb_1+\kb_2}\big). \nn 
  \end{align}

We could also compute $H_4$, but we will not need it.

\section{Effective Friedrichs Hamiltonian}
\label{Effective Friedrichs Hamiltonian}

Recall that $\Omega_\alpha =W^*_\alpha \Omega$.
Let  $\Omega_\Bog: = U_\Bog^*\Omega_\alpha$ be the quasiparticle
vacuum.  Let $\Span^\cl(K)$ denote the closure of the
  span of the set $K\subset\mathcal{F}$. Introduce the space   $\cF^\exc $ consisting of the Bogoliubov vacuum and 
  quasiparticle excitations, and its $n$-quasiparticle sector:
\begin{align*}
  \cF^\exc :
  =&\Span^\cl\{b_{\kb_1}^*\cdots b_{\kb_n}^*\Omega_\Bog\ |\      \kb_1,\dots,\kb_n\neq0,\quad n=0,1,\dots\},\\
\cF_n^\exc :=&
\Span^\cl\{b_{\kb_1}^*\cdots b_{\kb_n}^*\Omega_\Bog\ |\
\kb_1,\dots,\kb_n\neq0\}.\end{align*}
The most ``violent'' approximation that we are going to make is compressing the Hamiltonian $H$ into the space $\cF^\exc $. We
also drop the uninteresting constant $\kappa^{-1}H_0$ and the (somewhat
more interesting) constant $E_\mathrm{Bog}$. Thus we introduce the
{\em excitation Hamiltonian}
\be
H^\exc :=I^{\exc *}
\big(H-\kappa^{-1}H_0-E_\mathrm{Bog}\big)I^\exc , \nn \ee
where $I^\exc $ denotes the  embedding of $\cF^\exc $ in $\cF$.
Thus $H^\exc $ is an operator on $\cF^\exc $ and
\be
H^\exc =
 H_\Bog+\sqrt\kappa H_3^\exc +\kappa H_4^\exc ,\ee
where $H_3^\exc $  and $ H_4^\exc $ are defined in
\eqref{exc3} and \eqref{exc4}.

\begin{remark}
Let us make some remarks concerning the algebraic meaning of the above
construction. Our  physical space (in finite volume and in the
momentum representation) is the bosonic Fock space over the 1-particle space
$l^2\big(\frac{2\pi}{L}\mathbb{Z}\big)$. By the exponential property of
Fock spaces (see e.g. \cite{DeGe}) we have the following identification:
\begin{equation}
  \Gamma_\mathrm{s}\Big(l^2\big(\tfrac{2\pi}{L}\mathbb{Z}\big)\Big)\simeq
\Gamma_\mathrm{s}(\mathbb{C})\otimes
\Gamma_\mathrm{s}\Big(l^2\big(\tfrac{2\pi}{L}\mathbb{Z}\setminus\{0\}\big)\Big),\label{expon}
\end{equation}
where the first factor describes the ``zeroth mode'' treated as the
``condensate''
and the second
``excitations outside of the condensate''. We will denote by $U$ the
(unitary and canonical) identification described in \eqref{expon}. Note
that creation and annihilation operators of non-zero modes, $a_\kb^*,
a_\kb$, $\kb\neq0$, as well as of quasiparticles $b_\kb^*,b_\kb$ act
only in the second factor.
The translations also act only in the second factor. The coherent vector
$\Omega_\alpha=W_\alpha^*\Omega$ is translation invariant and can be understood as
an element of the first factor. Thus $U$ identifies $\cF^\exc$ with
\begin{equation}
 \Omega_\alpha \otimes
\Gamma_\mathrm{s}\Big(l^2\big(\tfrac{2\pi}{L}\mathbb{Z}\setminus\{0\}\big)\Big),\label{expon1}
\end{equation}
The compressed Hamiltonian $H^\exc$ can be then interpreted as
\begin{align}
  H^\exc=
 (\Omega_\alpha|\  UHU^*
  |\Omega_\alpha),\end{align}
which is an operator on $\Gamma_\mathrm{s}\Big(l^2\big(\tfrac{2\pi}{L}\mathbb{Z}\setminus\{0\}\big)\Big)$.

The idea of decomposing the Fock space as in \eqref{expon}, where the 
first factor describes the ``codensate'', is common in the 
literature. It is e.g.  used in \cite{DerNap-13}
and (implicitly) in the paper by Lewin-Nam-Serfaty-Solovej \cite{LewNamSerSol-15}.

Our compression construction is essentially the most direct
interpretation of the ``replacing the zeroth mode by  a
c-number'', which is a very common procedure in the physics literature.
Physicists expect that this procedure yields physically relevant
results. And so do we,  at least concerning the imaginary part of the
dispersion relation.

However, of course, compression produces an operator which is not unitarily 
equivalent to the initial operator. Therefore, it is certainly a
rather fishy step in our analysis: rigorously it is not clear how much the
analysis of $H^\exc$ will say  about $H$.
\end{remark}

We make two more approximations. 
 We drop 
$\kappa 
H_4$, which is of higher order in $\kappa$ than $\sqrt\kappa H_3$. We
also drop $H_{3,2}$, which involves $3$-quasiparticle
creation/annihilation operators, and does not contribute to the
damping rate (see Section \ref{sec:cubicirrelevant} for a justification). Thus $H^\exc$ is replaced with
\be
\label{exc5} H^\eff:=H_\Bog+\sqrt\kappa 
H_{3,1}^\exc.\ee
 To make our following discussion consistent with Sect. \ref{Fermi
  Golden Rule} about the Fermi Golden Rule,
 we introduce a new coupling constant 
    \be \label{def:couplinglambda}
    \lambda:=\sqrt\kappa. 
    \ee
Let $\kb\neq0$.   Clearly,  $b_\kb^*\Omega_\Bog$ is an eigenstate of
$H^\eff$ for  $\lambda=0$.
We would like to compute
the self-energy for the vector $b_\kb^*\Omega_\Bog$ and the
Hamiltonian  $H^\eff$:
\be \label{droppo1}
\lambda^2\Sigma^\eff_\kb (z):=\frac{-1}{(b_\kb^*\Omega_\Bog|(z-H^{\rm eff})^{-1} b_\kb^*\Omega_\Bog)}
+z-e_\kb.\ee 
 Introduce the  subspaces of
$\mathcal{F}^\mathrm{exc}$  and $\mathcal{F}_n^\mathrm{exc}$ with the total 
momentum $\kb$:
 \begin{align*}
  \mathcal{F}^\mathrm{exc}(\kb) &:=\Span^\cl\{b_{\kb_1}^*\cdots 
  b_{\kb_n}^*\Omega_\Bog,\qquad\kb_1+\cdots\kb_n=\kb, \
   \kb_1,\dots,\kb_n\neq0,\quad n=0,1,\dots\},\\
  \mathcal{F}_n^\mathrm{exc}(\kb) &:=\Span^\cl\{b_{\kb_1}^*\cdots 
  b_{\kb_n}^*\Omega_\Bog,\qquad\kb_1+\cdots\kb_n=\kb, \
   \kb_1,\dots,\kb_n\neq0\}.
   \end{align*}
$b_\kb^*\Omega_\Bog$ is contained in the space
 $\cF^\exc(\kb)$, which is preserved by $H^\eff$.
Let
$H^\eff(\kb)$ denote the operator $H^\eff$ restricted  to
$\cF^\exc (\kb)$.
Thus we can restrict ourselves to  the fiber space
$\cF^\exc(\kb)$ and the fiber Hamiltonian $H^\eff(\kb)$.
In particular, in \eqref{droppo1} we can replace 
 $H^\eff$ with  $H^\eff(\kb)$.

For simplicity, we
will assume that $\frac12\kb\not\in\frac{2\pi}{L}\mathbb{Z}$, so that
at least one  coordinate of $\kb$ is  odd. This guarantees that
$\pb\neq\kb-\pb$.
Let $Z_\kb^L$ denote the set of (unordered) pairs 
$\{\pb,\kb-\pb\}\subset\frac{2\pi}{L}\mathbb{Z}^3\setminus\{0,\kb\}$. Then 
$\cF_2^\exc(\kb)$ can be identified with $l^2(Z_\kb^L)$.

For our analysis it is enough to know only $H^\eff$  (or
$H^\eff(\kb)$) compressed 
  to $\cF_1^\exc(\kb)\oplus\cF_2^{\rm exc}(\kb)$. Note that
the one-quasiparticle state $b_\kb^*|\Omega_\Bog\rangle$
spans $\cF_1^\exc(\kb)$, and
  $\cF_2^{\rm exc}(\kb)$ is spanned by 
$b_\pb^* b_{\kb-\pb}^*\Omega_\Bog$ with   $\{\pb,\kb-\pb\}\in Z_\kb^L.$
We compute:
 \begin{theorem}
\begin{align}\label{eff1}
(b_\kb^*\Omega_\Bog|H^{\rm eff}b_\kb^*\Omega_\Bog)&=e_\kb ,\\ \label{eff2}
(b_\pb^*b_{\kb-\pb}^*\Omega_\Bog|H^{\rm eff}b_\pb^*b_{\kb-\pb}^*\Omega_\Bog)&=e_\pb+e_{\kb-\pb}
                                               ,\\
\label{eqn:overlap_hV}
(b_\pb^*b_{\kb-\pb}^*\Omega_\Bog|H^{\rm eff}
 b_\kb^*\Omega_\Bog)
     &=\frac{ \lambda}{L^{3/2}}h_\kb(\pb),\\
\label{eqn:overlap_hV.}
(b_\kb^*\Omega_\Bog |H^{\rm eff}b_\pb^*b_{\kb-\pb}^*\Omega_\Bog)
     &=\frac{ \lambda}{L^{3/2}}h_\kb(\pb)
\end{align}
with

  \begin{align}
    \label{eqn:h(p)}
h_\kb(\pb)&=  \sqrt{\frac{\mu \hat{v}^2(\kb) 
            }{\hat{v}(0)}}\big(\gamma_\kb-\sigma_\kb\big) 
            \big(\gamma_\pb\sigma_{\kb-\pb}+\sigma_\pb \gamma_{\kb-\pb}\big)\\
    &+ \sqrt{\frac{\mu \hat{v}^2(\pb) 
      }{\hat{v}(0)}}\big(\sigma_\kb\sigma_\pb\sigma_{\kb-\pb}-\gamma_\kb\gamma_\pb\gamma_{\kb-\pb}+\gamma_\kb\sigma_\pb\gamma_{\kb-\pb}
      -\sigma_\kb\gamma_\pb\sigma_{\kb-\pb}\big)\nn\\
          &+\sqrt{\frac{\mu \hat{v}^2(\kb-\pb) 
      }{\hat{v}(0)}}\big(\sigma_\kb\sigma_\pb\sigma_{\kb-\pb}-\gamma_\kb\gamma_\pb\gamma_{\kb-\pb}+\gamma_\kb\gamma_\pb\sigma_{\kb-\pb}
      -\sigma_\kb\sigma_\pb\gamma_{\kb-\pb}\big).\nn
  \end{align}
\end{theorem}

\proof \eqref{eff1} and \eqref{eff2} are straightforward. Let us prove
\eqref{eqn:overlap_hV}. We have
\begin{align}
  &(b_\pb^*b_{\kb-\pb}^*\Omega_\Bog|H^{\rm eff}
 b_\kb^*\Omega_\Bog)
\\=&
(b_\pb^*b_{\kb-\pb}^*\Omega_\Bog|H_{3,1}^\exc
 b_\kb^*\Omega_\Bog).\label{eff3}
\end{align}
Remembering that we have $\pb\neq\kb-\pb$, we see that the only terms
in $H_{3,1}^\exc$ which contribute to \eqref{eff3} are
\begin{align}
  b_{\pb}^*b_{\kb-\pb}^*  b_{\kb}
  \Bigg(  \frac{\sqrt{\mu}\hat
  v(\pb)}{L^{3/2}\sqrt{\hat v(0)}}
   &\big(\sigma_{\kb}\sigma_{\pb}\sigma_{\kb-\pb}-\gamma_{\kb}\gamma_{\pb}\gamma_{\kb-\pb}+ \gamma_{\kb}\sigma_{\pb}\gamma_{\kb-\pb}-\sigma_{\kb}\gamma_{\pb}\sigma_{\kb-\pb} \big)\nn \\
+ & \frac{\sqrt{\mu}\hat
  v(\kb)}{L^{3/2}\sqrt{\hat v(0)}} \big(
    \gamma_{\kb}-\sigma_{\kb}\big)\gamma_{\pb}\sigma_{\kb-\pb}\\
  + 
  \frac{\sqrt{\mu}\hat
  v(\kb-\pb)}{L^{3/2}\sqrt{\hat v(0)}}
   &\big(\sigma_{\kb}\sigma_{\pb}\sigma_{\kb-\pb}-\gamma_{\kb}\gamma_{\pb}\gamma_{\kb-\pb}+ \gamma_{\kb}\gamma_{\pb}\sigma_{\kb-\pb}-\sigma_{\kb}\sigma_{\pb}\gamma_{\kb-\pb} \big)\nn \\
+ & \frac{\sqrt{\mu}\hat
  v(\kb)}{L^{3/2}\sqrt{\hat v(0)}} \big(
    \gamma_{\kb}-\sigma_{\kb}\big)\sigma_{\pb}\gamma_{\kb-\pb}\Bigg)
\end{align}
This yields \eqref{eqn:overlap_hV}. 
\qed

The Hamiltonian $H^\eff$ compressed to
$\cF_1^\exc(\kb)\oplus\cF_2^{\rm exc}(\kb)$ will be 
called  the {\em effective Friedrichs Hamiltonian}
  (for volume $L^3$ and momentum $\kb$). It is
denoted
$  H_{\rm{Fried}}^{L}(\kb)$ and given by 
\begin{align}
\label{def:themodel1}
 H_{\rm{Fried}}^{L}(\kb)&:=\begin{bmatrix} 
 	e_\kb & \frac{\lambda}{L^{3/2}}(h_{\bf k}|\\
 	\frac{\lambda}{L^{3/2}}|h_{\bf k})& e_\pb+e_{\kb-\pb}
      \end{bmatrix},\\
\text{on }\quad  \cF_1^\exc(\kb)\oplus\cF_2^{\rm exc}(\kb)&\simeq
                                                \C\oplus
       l^2(Z_\kb^L),\end{align}
                                               where we explicitly introduced a reference to the volume $L^3$ in 
the notation. 
Thus we end up in a situation described in Section \ref{Fermi Golden
  Rule},  with $b_\kb^*\Omega_\Bog$, resp. $
  l^2(Z_\kb^L)$ corresponding to $\Psi$, resp. $\cK$.    According to the Fermi Golden Rule  \eqref{fgr}
 the self-energy of $H_{\rm{Fried}}^{L}(\kb)$ is
\begin{align} \label{droppo3} 
\Sigma_\kb^L(z)&=\frac{1}{2 L^3}\sum_{\pb,\kb-\pb\neq0}\frac{h_{\bf k}^2(\pb)}{(z-e_\pb-e_{\kb-\pb})},
\end{align}
where $\frac12$ in front of the sum accounts for double
  counting.

The function $\pb\mapsto e_\pb$ is well defined for all $\pb\in\R^3$,
and not only for $\pb \in\frac{2\pi}{L}\Z^3\setminus\{0\}$. 
Similarly, $h_\kb(\pb)$ are well defined for
all $\pb\in\R^3\setminus\{0,\kb\}$, and not only for
$\frac{2\pi}{L}\Z^3\setminus\{0,\kb\}$. The expression
    \eqref{droppo3}  can be interpreted as the Riemann sum converging
    as $L\to\infty$ to the integral
    \begin{align}
       \Sigma_\kb(z)&=\frac{1}{2(2\pi)^3}\int\frac{h_\kb(\pb)^2\d\pb}{(z-e_\pb-e_{\kb-\pb})}. \label{droppo4} 
\end{align}
We can also introduce the {\em infinite volume effective Friedrichs
Hamiltonian}
\begin{equation}
\begin{aligned}
\label{def:themodela}
 H_{\rm{Fried}}(\kb)&:=\begin{bmatrix} 
 	e_\kb & \lambda(h_\kb|\\
 	\lambda|h_\kb) & e_\pb+e_{\kb-\pb}
      \end{bmatrix},\\
\text{on }\quad  & \C\oplus L^2(\R^3/\Z_2),\end{aligned}
\end{equation}
where $\Z_2$ is the two-element group generated by
$\pb\mapsto\kb-\pb$.
The Fermi Golden Rule predicts that
$\Sigma_\kb(e_\kb+ \ri0)$ describes the energy shift of the
eigenvalue of the infinite volume  Hamiltonian  $
H_{\rm{Fried}}(\kb)$.

It is maybe worth mentioning that all the steps that lead to $
H_{\rm{Fried}}^L(\kb)$ and $
H_{\rm{Fried}}(\kb)$ are translation invariant.

\section{The shape of the quasiparticle spectrum}
\label{sec:cubicirrelevant}

If $\kb\mapsto e_\kb$ is a dispersion relation of quasiparticles, then the infimum of
the $n$-quasiparticle spectrum is
\be\inf\{e_{\pb_1}+\cdots e_{\pb_n}\ |\
\pb_1+\cdots+\pb_n=\kb\}.\label{n-quasi}\ee 
Sometimes, it is possible to compute \eqref{n-quasi} exactly, as shown
in the following lemma.

\begin{lemma} \label{lem:convex_shape} Let $\kb\mapsto e_\kb$ be a convex function. Then
 \be \label{def:infconv}
 \inf_{\pb}\{e_\pb+e_{\kb-\pb}\}= 2 e_{\kb/2}.
 \ee
 In particular, \be\inf_{\pb}\{e_\pb+e_{\kb-\pb}\}\leq e_{\kb}. \label{conv}\ee
 If in addition $\kb\mapsto e_\kb$ is a strictly  convex function, then
 \be\inf_{\pb}\{e_\pb+e_{\kb-\pb}\}< e_{\kb},\quad \kb\neq0. \label{conv1}\ee
 \end{lemma}
 \begin{proof}
The left hand side of \eqref{def:infconv} is called infimal involution and is often denoted as 
\be 
e \square e ( \kb ):= \inf_{\pb}\{e_\pb+e_{\kb-\pb}\}.
\ee
Since \( e_\kb\) is a convex function so is $e \square e ( \kb )$  \cite[Chapter 12]{BauCom-17} and it satisfies 
\be 
(e \square e) ^*=e^*+e^*=2e^*
\ee
where $e^*$ denotes the   Legendre–Fenchel transform of $e$.  Hence 
\begin{eqnarray*}
   \inf_{\pb}\{e_\pb+e_{\kb-\pb}\} = e \square e ( \kb )
                               = (e \square e)^{**}( \kb )
                                = (2 e^*)^*( \kb )
                                = 2 e_{\kb/2}  
\end{eqnarray*}
which proves \eqref{def:infconv}. Now \eqref{conv}
follows from convexity. Indeed,
$$ 2e_{\pb/2}= 2 e_{\pb/2+0/2}\le  e_{\pb}.$$
\end{proof}

Now  $e_{\kb}$  in \eqref{coeff:bogdisp2}, that is
\be \label{strictly}
e_\kb= \sqrt{\frac{1}{4}|\kb|^4+
          \mu|\kb|^2},\ee is strictly
convex. Therefore, \eqref{conv1} is true, and so the dispersion
relation is embedded inside the 2-quasiparticle spectrum.

\begin{remark} If we replace Assumption \eqref{assu-flat}
  with Assumption
  \eqref{assu-generic}, and suppose
  \be 1+2\mu\frac{\nu}{\hat v(0)}>0,\label{assu-gen2}\ee
  then  the dispersion
  relation is still
  embedded inside the  unperturbed 2-quasiparticle spectrum, at least for
small momenta. The same is true for the  effective Friedrichs
Hamiltonian $H_\mathrm{Fried}(\kb)$ for small $\kb$.\end{remark}

The Hamiltonian $H^\mathrm{exc}$ couples $b_\kb^*\Omega_\Bog$ with
4-quasiparticle states through $H_{3,2}^\exc$. The bottom of
4-quasiparticle spectrum lies below the dispersion relation (in fact,
if it is given by \eqref{coeff:bogdisp2}, it is equal to
$4e_{\kb/4}<e_\kb$). However, $H_{3,2}^\exc$ does not couple 
$b_\kb^*\Omega_\Bog$ to all possible 4-quasiparticle states with the
total momentum $\kb$, but only to states of the form
$b_{\pb_1}b_{\pb_2}b_{\pb_3}b_\kb\Omega_\Bog$ with
$\pb_1+\pb_2+\pb_3=0$. Their energy 
is
\be e_\kb+e_{\pb_1}+e_{\pb_2}+e_{\pb_3}\geq e_\kb.\ee
Thus the state
$b_\kb^*\Omega_\Bog$ is situated at the boundary of the
energy-momentum spectrum and the only coupling is through $\pb_1=\pb_2=\pb_3=0$.
Before going to the thermodynamic limit this is excluded, because on the
excited space all momenta are different from  zero. Assuming that this
effect survives the thermodynamic limit, we expect that the term
$H^\exc_{3,2}$ does not lead to damping and we therefore drop it
from $H_\mathrm{Fried}$, even though
in terms of the coupling parameter $\kappa$ this term  is of the same order as $H^\exc_{3,1}$, which we keep in our analysis.

 Two-quasiparticle states are coupled to
  three-quasiparticle states through $H_{31}^\exc$ and to
five-quasiparticle states through $H_{32}^\exc$. These couplings,
however, do not contribute to our Fermi Golden Rule computation---they
affect the damping rate in a higher order of the coupling
constant. Therefore, we do not include these states in our Hilbert space
$\cF_1^\exc(\kb)\oplus\cF_2^{\rm exc}(\kb)$ on which our effective
Friedrichs Hamiltonian acts.

\section{Computing the self-energy}

In the remaining part of our paper, the main goal will be to compute
approximately the 3-dimensional
integral \eqref{droppo4}. To do this efficiently it
is important to choose a convenient coordinate system.

Let us introduce the notation $k=|\kb|$,  $p=|\pb|$,  $l=|\lb|$, where $\lb=\kb-\pb$. 
One could try to compute \eqref{droppo4} using the spherical
coordinates for $\pb$ with respect to the axis determined by $\kb$. This means using $p=|\pb|,w=\cos\theta,\phi$, so that
$\pb=(p\sqrt{1-w^2}\cos\phi, p\sqrt{1-w^2}\sin\phi,pw)$. 
The self-energy in these coordinates is
 \begin{align}
       \Sigma_\kb(z)&=\frac{1}{2 (2\pi)^3}\int_0^\infty\int_{-1}^1\int_0^{2\pi}\frac{h_\kb(p,w)^2p^2\d 
                              p\d w\d\phi}{(z-e_p-e_{l(p,w)})} \label{droppo5.} 
\end{align}
where, with abuse of notation, $h_\kb(p,w)$ is the function $h_\kb(\pb)$ in the variables $p, w,\phi$. The variable $\phi$ can be easily integrated out. $h_\kb(\pb)$
depends only on $k,p,l$ and (\ref{droppo5.}) can be
rewritten as
 \begin{align*}
       \Sigma_{\kb}(z)&=\frac{1}{2 (2\pi)^2}\int_0^\infty\int_{-1}^1\frac{(h_k(p,l(p,w)))^2p^2\d 
                              p\d w}{(z-e_p-e_{l(p,w)})}, 
 \end{align*}

The coordinates  $p,w$ are not convenient because they break the
natural symmetry $\pb\to \kb-\pb$ of the system. 
Instead of $p,w$ it is much better to use the
variables
$p,l$. Note the constraints
\begin{align} |p-l|&\leq k,\label{triangle1}\\
  k&\leq p+l,\label{triangle2}\end{align}
that follow from the triangle inequality.
We have $w=\frac{k^2+p^2-l^2}{2kp}$.
The Jacobian is easily computed:
\be p^2\d p\d w=\frac{pl}{k}\d p\d l=\frac{1}{4k}\d p^2\d l^2.\ee
Let us make another change of variables:
\be 
t= p+l, \quad  s=p-l;\qquad p=\frac{t+s}{2},\quad l=\frac{t-s}{2}; 
\ee 
\be \d p^2\d l^2=\frac{t^2-s^2}{2}\d t\d s.\ee 
The limits of integration following from the constraints
\eqref{triangle1} and  \eqref{triangle2} are very easy
to impose:
\begin{align} \label{eq:Sigma_t_l}
       \Sigma_{\kb}(z)&=\frac{1}{2 (2\pi)^2}\int_k^\infty\d
                                          t\int_{-k}^k\d s\frac{h_k(t,s)^2
                      (t^2-s^2)                   }{8k(z-e_{\frac{t+s}{2}}-e_{\frac{t-s}{2}})},  
\end{align}

Another choice of variables can also be useful.
If $k\mapsto e_k$ is an increasing function, which is always the case
for small $k$, but also for the important case of constant $\frac{\hat
  v(\kb)}{\hat v(0)}$,
we can use the variables $u:=e_p$ and $w:=e_l$.
Set
\be f(e_k):=\frac {\d k^2}{\d e_k^2}.\ee

Thus we change the variables
\be 
\frac{1}{4k}\d p^2\d l^2=\frac{1}{4k} f(u)f(w)\d u^2\d w^2.\ee
 \begin{align*}
       \Sigma_{\kb}(z)&=\frac{1}{2 (2\pi)^2}\int\frac{h_{k}(u,w)^2f(u)f(w)\d
                                          u^2\d w^2}{4k(z-u-w)}, 
\end{align*}

We then perform a further change of variable 
\be 
x= u+w, \quad  y=u-w;\qquad u=\frac{x+y}{2},\quad w=\frac{x-y}{2}; 
\ee 
\be \d u^2\d w^2=\frac{x^2-y^2}{2}\d x\d y.\ee

Now we can write
\begin{align*}
	\Sigma_{\kb}(z) 
	&= \frac{1}{16 \pi^2 k }\iint\frac{h_{k}(x,y)^2 f(\frac{x+y}{2})f(\frac{x-y}{2}) (x^2-y^2)\d y \d x}{4(z-x)}, 
\end{align*}
where the limits of integration are somewhat more difficult to
describe.

When $\frac{\hat v(\kb)}{\hat v(0)}$ is a constant, so that
\be e_k=k\sqrt{\mu+\frac{k^2}{4}},\qquad k^2=2\big(\sqrt{e_k^2+\mu^2}-\mu\big),\label{consta}\ee
we can compute the function $f$:
\be
f(u)=\frac{1}{\sqrt{u^2+\mu^2}}.
\ee
We also have
\begin{align}
              \sigma_{k}&=\sqrt{\frac{\frac{{k}^2}{2}+\mu+\sqrt{\frac{{k}^4}{4}+\mu{k}^2}}{2\sqrt{\frac{{k}^4}{4}+\mu{k}^2}}},\quad
                                        \gamma_{k}=\sqrt{\frac{\frac{{k}^2}{2}+\mu-\sqrt{\frac{{k}^4}{4}+\mu{k}^2}}{2\sqrt{\frac{{k}^4}{4}+\mu{k}^2}}}.
\end{align}

\section{Damping rate}\label{sec:dampingrate}

The following theorem is the main result of this paper.
\begin{theorem}\label{thm:damping}
 Suppose that the potential satisfies Assumption
  \eqref{assum}. Then
\be\label{damping1}
   \Sigma_{\kb} 
  (e_k+\ri0)=-c_\mathrm{Bel}k^5+
O(k^6)  \qquad \text{as} \qquad k\to 0,\qquad  c_\mathrm{Bel}=\frac{3\hat{v}(0)}{640\pi \mu}. 
  \ee
  \end{theorem}

\begin{remark} If we replace Assumption \eqref{assu-flat} with Assumption   \eqref{assu-generic} with $\nu=0$, then Theorem \ref{thm:damping} remains true.
\end{remark}

  \begin{proof}[Proof of Theorem \ref{thm:damping}]
We will use the variables $x,y$:
\begin{align} 
	\Sigma_{\kb}(e_k+\ri0)
	&= \frac{1}{16 \pi^2 k }\iint\frac{h_{\kb}(x,y)^2(x^2-y^2)\d y \d x}{(e_k-x+\ri 0)\sqrt{(x+y)^2+4\mu^2}\sqrt{(x-y)^2+4\mu^2}}. \label{droppo7} 
\end{align}
It follows from (\ref{droppo7})  and the Sochocki-Plemelj  formula \eqref{eq:sochocki} that 
\begin{align}
\Sigma_{\kb}(e_k+\ri0)&=\mathrm{Re}\Sigma_{\kb}(e_k+\ri0)+\ri\mathrm{Im}\Sigma_{\kb}(e_k+\ri0),\nn \\
\mathrm{Re}\Sigma_{\kb}(e_k+\ri0)	&= \frac{1}{16 \pi^2 k }\iint\frac{h_{\kb}(x,y)^2(x^2-y^2)\d y \d x}{(e_k-x)\sqrt{(x+y)^2+4\mu^2}\sqrt{(x-y)^2+4\mu^2}}\\
\mathrm{Im}\Sigma_{\kb}(e_k+\ri0) 		&=-\frac{ 
    \pi}{16 \pi^2 k} \iint \frac{h_{\kb}(x,y)^2 (x^2-y^2) \delta(e_k-x)\d 
		y\d x}{\sqrt{(x+y)^2+4\mu^2}\sqrt{(x-y)^2+4\mu^2}}
        \label{eqn:int-selfenergy} \\
  		&=-\frac{ \pi}{16 \pi^2 k} \int
           \frac{h_{\kb} (e_k,y)^2 (e_k^2-y^2) \d y}{\sqrt{(e_k+y)^2+4\mu^2}\sqrt{(e_k-y)^2+4\mu^2}} .    \label{eqn:int-selfenergy1}
\end{align}

    Our starting point is the expression \eqref{eqn:int-selfenergy1}. Obviously, we first need to establish the integration limits in $y$. Recall that $y=e_p -e_l$ but under the additional constraint that $e_k =e_p +e_l$  which comes from the constraint $\delta(x-e_k)$ in \eqref{eqn:int-selfenergy}. It follows immediately that
    $-e_k \leq y\leq e_k$. 
      Thus, for $|\kb|$ small enough we can replace $\hat v$ with
      $\hat v(0)$, so that
  \begin{align}
\label{eqn:h(p')}
h_\kb(\pb)&= 2 \sqrt{\mu \hat{v}(0)}
            \Big(\sigma_\pb\gamma_{-\kb}\gamma_{\pb-\kb}+\sigma_{\kb-\pb}\gamma_{-\kb}\gamma_\pb+\sigma_\pb\sigma_{\kb-\pb}\sigma_\kb\\
  &\qquad- \gamma_\pb\sigma_{-\kb}\sigma_{\pb-\kb}-\gamma_{\kb-\pb}\sigma_{-\kb}\sigma_\pb- \gamma_\pb\gamma_{\kb-\pb}\gamma_\kb \Big).\nn
\end{align}
Hence 
\begin{align}
	\frac{h_{\kb}(\pb)}
	 {2 \sqrt{\mu \hat{v}(0)}}
	&=\sigma_k(\sigma_p\sigma_l-\sigma_l\gamma_p-\sigma_p\gamma_l)+\gamma_k(\sigma_p\gamma_l+\sigma_l\gamma_p-\gamma_p\gamma_l).\nn\\
	&= \frac{\sigma_k}{2\sqrt{uw}}\bigg( \sqrt{\sqrt{u^2+\mu^2}+u}\sqrt{\sqrt{w^2+\mu^2}+w}-\sqrt{\sqrt{w^2+\mu^2}+w}\sqrt{\sqrt{u^2+\mu^2}-u}    \nn  \\
	&- \sqrt{\sqrt{u^2+\mu^2}+u}\sqrt{\sqrt{w^2+\mu^2}-w}\bigg)\nn\\
	&+ \frac{\gamma_k}{2\sqrt{uw}}\bigg( \sqrt{\sqrt{u^2+\mu^2}+u}\sqrt{\sqrt{w^2+\mu^2}-w}+\sqrt{\sqrt{w^2+\mu^2}+w}\sqrt{\sqrt{u^2+\mu^2}-u}    \nn  \\
	&- \sqrt{\sqrt{u^2+\mu^2}-u}\sqrt{\sqrt{w^2+\mu^2}-w}\bigg)\\
	&= \frac{1}{2\sqrt{x^2-y^2}}\bigg( \sigma_k \sqrt{(A_1+x+y)(A_2+x-y))}-\gamma_k\sqrt{(A_1-x-y)(A_2-x+y)}    \nn  \\
	&+ (\gamma_k-\sigma_k)\sqrt{(A_1-x-y)(A_2+x-y))}+(\gamma_k-\sigma_k)\sqrt{(A_1+x+y)(A_2-x+y))} \bigg),
\end{align}
where 
\be \label{eqn:A1A2}
A_1:=A_1(x,y)= \sqrt{(x+y)^2+4\mu^2}, \qquad A_2:= A_2(x,y)= \sqrt{(x-y)^2+4\mu^2}.
\ee
Therefore the integrand in \eqref{eqn:int-selfenergy1} becomes 
\begin{align}
	& \frac{(h_{\kb}(x,y))^2 (x^2-y^2)}{\sqrt{(x+y)^2+4\mu^2}\sqrt{(x-y)^2+4\mu^2}}\\
	&= \frac{\mu\hat{v}(0)}{A_1A_2}\bigg( \sigma_k \sqrt{(A_1+x+y)(A_2+x-y))}-\gamma_k\sqrt{(A_1-x-y)(A_2-x+y)}    \nn  \\
	&+ (\gamma_k-\sigma_k)\sqrt{(A_1-x-y)(A_2+x-y))}+(\gamma_k-\sigma_k)\sqrt{(A_1+x+y)(A_2-x+y))} \bigg)^2. \nn \\
 &= \frac{\mu \hat{v}(0)}{A_1A_2}\bigg(\sigma_k^2\left(  3 A_1A_2+(x+y)A_2+(x-y)A_1-(x^2-y^2)-4\mu(A_1+A_2+2x)+8\mu^2\right)\nn \\
	&+\gamma_k^2\left(  3 A_1A_2-(x+y)A_2-(x-y)A_1-(x^2-y^2)-4\mu(A_1+A_2-2x)+8\mu^2\right)\nn \\
	&+ 2 \sigma_k \gamma_k \left( 4\mu A_1+4\mu A_2-2A_1A_2+2(x^2-y^2)-12\mu^2\right) \bigg). \label{droppo11}
\end{align}

Thus the equation in (\ref{eqn:int-selfenergy1}) becomes
\begin{align}
	& -\frac{ 1}{ 16 \pi k}\int_{-e_k}^{e_k} \d y \frac{h_{\kb}^2(x,y)(x^2-y^2)}{\sqrt{(x+y)^2+4\mu^2}\sqrt{(x-y)^2+4\mu^2}}\\
	&= \left( -\frac{ \mu\hat{v}(0)}{ 16 \pi k}\right) \int_{-e_k}^{e_k} \d y \bigg( \left( 3\sigma_k^2+3\gamma_k^2-4\sigma_k\gamma_k \right) + (\sigma_k^2-\gamma_k^2) \left( \frac{x-y}{A_2}+ \frac{x+y}{A_1}-\frac{8\mu x }{A_1A_2}\right)     \nn  \\
	&+ (-\sigma_k^2-\gamma_k^2+4\sigma_k\gamma_k ) \frac{x^2-y^2}{A_1A_2} -4\mu (\sigma_k-\gamma_k)^2\frac{A_1+A_2}{A_1A_2} + 8\mu^2(\sigma_k^2+\gamma_k^2-3\sigma_k\gamma_k )\frac{1}{A_1A_2}\bigg).
	 \label{droppo12}
\end{align}
The integrals involving $\frac{x\pm y}{A_j}$ and $\frac{1}{A_j}$ (where $j=1,2$ ) can be computed explicitly. In particular, setting $x=e_k$,  it follows that for \(j=1,2\)
 \begin{align} \label{expt-intgl-dropp011-a}
&\int_{-e_k}^{e_k} \d y \frac{e_k\pm y}{A_j (e_k,y)}=\int_{-e_k}^{e_k} \d y \left( \frac{e_k\pm y}{\sqrt{(e_k\pm y)^2+4\mu^2}}\right) 
=  2\sqrt{\mu^2+e_k^2}-2\mu, \\
& \label{expt-intgl-dropp011-c}
\int_{-e_k}^{e_k} \d y \frac{1}{A_j (e_k,y)}=  \int_{-e_k}^{e_k} \d y \left( \frac{1}{\sqrt{(e_k\pm y)^2+4\mu^2}}\right) 
= \log \left(\frac{e_k}{\mu}+ \sqrt{1+\frac{e_k^2}{\mu^2}} \right).
\end{align}
This yields 
\begin{align}
	&
	\left( -\frac{ 1}{ 16 \pi k}\right) \int_{-e_k}^{e_k}\d y \left( \frac{h^{2}_{\kb}(e_k,y) (e_k^2-y^2)}{\sqrt{(e_k+y)^2+4\mu^2}\sqrt{(e_k-y)^2+4\mu^2}}\right) \\
	&=\left( -\frac{ \mu\hat{v}(0)}{ 16 \pi k}\right) \left( 2\left( 3\sigma_k^2+3\gamma_k^2-4\sigma_k\gamma_k \right)e_k +  4\sqrt{\mu^2+e_k^2}-4\mu-8\mu (\sigma_k-\gamma_k)^2  \log \left(\frac{e_k}{\mu}+ \sqrt{1+\frac{e_k^2}{\mu^2}} \right)\right)\nn \\ 
	&+ \left( -\frac{ \mu\hat{v}(0)}{ 16 \pi k}\right) \int_{-e_k}^{e_k} \d y \bigg(   \frac{-(\sigma_k^2-4\sigma_k\gamma_k+\gamma_k^2) (e_k^2-y^2)-8\mu e_k+8\mu^2(\sigma_k^2+ \gamma_k^2-3\sigma_k\gamma_k)}{A_1A_2}\bigg)\label{intgl-dropp011-}.
\end{align} 
where two types of integrals, namely
\be 
\int \left( \frac{-y^2}{A_1A_2}\right) \d y \,\,\, \text{and}\,\,\,
\qquad \int \left( \frac{1}{A_1A_2}\right) \d y,
\ee
still appear as they cannot be computed explicitly. We will approximate them by expansions in \(e_k\)(which is small, as $k$ is small). To this end,  we recall
\be 
\sigma_k=\sqrt{\frac{\sqrt{e_k^2+\mu^2}+e_k}{2 e_k}}, \qquad \gamma_k=\sqrt{\frac{\sqrt{e_k^2+\mu^2}-e_k}{2 e_k}}, 
\ee
which gives
\be 
\sigma_k^2+\gamma_k^2= \frac{\sqrt{e_k^2+\mu^2}}{e_k},\qquad \sigma_k\gamma_k=\frac{\mu }{2 e_k}.
\ee 
Then (\ref{intgl-dropp011-})  equals to 
\begin{align}
	& \left( -\frac{ \mu\hat{v}(0)}{ 16 \pi k}\right) \left( 2\left( 3\sigma_k^2+3\gamma_k^2-4\sigma_k\gamma_k \right)e_k +  4\sqrt{\mu^2+e_k^2}-4\mu-8\mu (\sigma_k-\gamma_k)^2  \log \left(\frac{e_k}{\mu}+ \sqrt{1+\frac{e_k^2}{\mu^2}} \right)\right)\nn \\ 
	&+ \left( -\frac{ \mu\hat{v}(0)}{ 16 \pi k}\right) \int_{-e_k}^{e_k} \d y \bigg(   \frac{-(\sigma_k^2-4\sigma_k\gamma_k+\gamma_k^2) (e_k^2-y^2)-8\mu e_k+8\mu^2(\sigma_k^2+ \gamma_k^2-3\sigma_k\gamma_k)}{A_1A_2}\bigg) \nn\\
	&=\left( -\frac{ \mu\hat{v}(0)}{ 16 \pi k}\right) \left(  2(3\sqrt{e_k^2+\mu^2}-2 \mu) + 2 (2\sqrt{\mu^2+e_k^2}-2\mu)- 8 \mu \frac{\sqrt{e_k^2+\mu^2}-\mu}{e_k}\log \left(\frac{e_k}{\mu}+ \sqrt{1+\frac{e_k^2}{\mu^2}} \right) \right) \nn\\
	&+\left( \frac{ \mu\hat{v}(0)}{ 16 \pi k}\right) \frac{\sqrt{e_k^2+\mu^2}-2\mu}{e_k}\int_{-e_k}^{e_k} \d y \left(  \frac{ e_k^2-y^2}{A_1A_2}\right) \nn\\
	&+\left( \frac{ 1}{ 16 \pi k}\right) \left(  8 \mu^2 \hat{v}(0)e_k-8 \mu^3\hat{v}(0)\frac{2\sqrt{e_k^2+\mu^2}-3\mu}{2e_k} \right) \int_{-e_k}^{e_k} \d y \left(  \frac{ 1}{A_1A_2}\right) \nn\\
	&=\left( -\frac{ \mu\hat{v}(0)}{ 16 \pi k}\right) \left(  10\mu \sqrt{(e_k/\mu)^2+1} -8\mu- 8 \mu \frac{\sqrt{(e_k/\mu)^2+1}-1}{e_k/\mu }\log \left(\frac{e_k}{\mu}+ \sqrt{1+\frac{e_k^2}{\mu^2}} \right) \right) \label{intable-part}\\
	&+\left( \frac{ \mu\hat{v}(0)}{ 16 \pi k}\right) \frac{\sqrt{(e_k/\mu)^2+1}-2}{e_k/\mu}\int_{-e_k}^{e_k} \d y \left(  \frac{ e_k^2-y^2}{A_1A_2}\right) \label{nonintable-part-1}\\
	&+\left( \frac{ \mu\hat{v}(0)}{ 16 \pi k}\right) \left( \frac{ 8\mu  e_k^2-4 \mu^3(2\sqrt{(e_k/\mu)^2+1}-3)}{e_k} \right) \int_{-e_k}^{e_k} \d y \left(  \frac{ 1}{A_1A_2}\right)   \label{nonintable-part-2}.
\end{align} 
We expand \eqref{intable-part} up to order $O(e_k^8)$. A tedious computation yields
\be
\eqref{intable-part}=\left( -\frac{ \mu\hat{v}(0)}{ 16 \pi k}\right) \left( 2\mu + \frac{e_k^2}{\mu}+\frac{5e_k^4}{12\mu^3}-\frac{41e_k^6}{120\mu^5}+ O(e_k^8)\right) . \label{expan-intable-part}
\ee
We shall now deal with the terms \eqref{nonintable-part-1} and \eqref{nonintable-part-2}. To this end we write
\begin{align}
	A_1A_2 &=\sqrt{4\mu^2+(e_k+y)^2}\sqrt{4\mu^2+(e_k-y)^2}\\
	&= 4\mu^2 \sqrt{1+ \left(\frac{e_k+y}{2\mu} \right)^2 }\sqrt{1+ \left(\frac{e_k-y}{2\mu} \right)^2 }\\
	&= 4\mu^2\sqrt{1+\frac{e_k^2+ y^2}{2\mu^2}+ \left( \frac{e_k^2-y^2}{4\mu^2}\right) ^2}\\
	&= 4\mu^2\sqrt{1+Q_1}\\
	&= 4\mu^2  \left(1+ \frac{1}{2}Q_1-\frac{1}{8}Q_1^2 
	+ \frac{1}{16}Q_1^3 \right) + O(Q_1^4).
\end{align}
where 
\be 
Q_1:= \frac{e_k^2+ y^2}{2\mu^2}+ \left( \frac{e_k^2-y^2}{4\mu^2}\right) ^2
\ee
Then 
\be 
\frac{1}{A_1A_2}=\frac{1}{4\mu^2 (1+Q_2)}= \frac{1}{4\mu^2} (1- Q_2+Q_2^2-Q_2^3)+ O(Q_2^4)
\ee 
where 
\be 
Q_2:= \frac{1}{2}Q_1-\frac{1}{8}Q_1^2 
+ \frac{1}{16}Q_1^3 .
\ee 
This leads to 
\be \label{expand-1/A1A2}
\frac{1}{A_1A_2}= \frac{1}{4\mu^2}-\frac{e_k^2}{16\mu^4}+\frac{e_k^4}{64\mu^6}-\frac{e_k^6}{256\mu^8}-\frac{y^2}{16\mu^4}+\frac{e_k^2y^2}{16\mu^6}- \frac{9e_k^4y^2}{256\mu^8}+ \frac{y^4}{64\mu^6}- \frac{9e_k^2y^4}{256\mu^8}-\frac{y^6}{256\mu^8} + O (e_k^{\iota_1}y^{\iota_2})
\ee 
where \( \iota_1+\iota_2=7\). In turn
\be 
\int_{-e_k}^{e_k} \frac{1}{A_1A_2}\d y
= \frac{e_k}{2\mu^2}- \frac{e_k^3}{6\mu^4}+\frac{19e_k^5}{240\mu^6}-\frac{13 e_k^7}{280\mu^8} + O(e_k^8)
\ee
and 
 \be 
 \int_{-e_k}^{e_k} \frac{e_k^2-y^2}{A_1A_2}\d y
 =  \frac{e_k^3}{3\mu^2}-\frac{e_k^5}{10\mu^4}+\frac{11 e_k^7}{280\mu^6} + O(e_k^8).
 \ee
 This implies  
 \be 
 \eqref{nonintable-part-1}=\left( \frac{ \mu\hat{v}(0)}{ 16 \pi k}\right) \left( -\frac{e_k^2}{3\mu}+\frac{4e_k^4}{15\mu^3}-\frac{11e_k^6}{84\mu^5}\right)+ O(e_k^8),  \label{expan-nonintable-part-1}
 \ee 
 and 
 \be
 \eqref{nonintable-part-2}=\left( \frac{ \mu\hat{v}(0)}{ 16 \pi k}\right)  \left(2\mu +\frac{4e_k^2}{3\mu}+\frac{3e_k^4}{20\mu^3}-\frac{2e_k^6}{7\mu^5} \right)  +O(e_k^8). \label{expan-nonintable-part-2}
 \ee
 Combining  \eqref{expan-nonintable-part-1}, \eqref{expan-nonintable-part-2} and \eqref{expan-intable-part} we obtain 
 \begin{align}
 	& -\frac{1}{ 16 \pi k }\int_{-e^k}^{e_k}\frac{(h^{\Lambda}_{\kb}(e_k,y))^2(e_k^2-y^2)}{\sqrt{(e_k+y)^2+4\mu^2}\sqrt{(e_k-y)^2+4\mu^2}}\d y \nn \\
 	&=\left( -\frac{ \mu\hat{v}(0)}{ 16 \pi k}\right)  \left( \frac{5}{12}-\frac{41}{120}\right) \frac{e_k^6}{\mu^5}= -\frac{3\hat{v}(0)}{ 640 \pi\mu^4}\frac{e_k^6}{k}.
 	\label{droppo15}
 \end{align}
This yields \eqref{damping1}. 
\end{proof}

                                      \section{  Full self-energy} \label{sec:renorm}

Recall that the self-energy is given by
\begin{align} \label{eq:Sigma_t_l.}
       \Sigma_{\kb}(z)&=\frac{1}{ 2 (2\pi)^2}\int_k^\infty\d
                                          t\int_{-k}^k\d s\frac{h_k(t,s)^2
                      (t^2-s^2)                   }{8k(z-e_{\frac{t+s}{2}}-e_{\frac{t-s}{2}})},
\end{align}
   For contact potentials $h_k(\pb)$, that is, if $\hat v(\kb)=\hat
   v(0)$ for all $\kb$, then $h(\kb)$ is given by
\begin{align}\label{below}
	\frac{h_{\kb}(\pb)}
	 {2 \sqrt{\mu \hat{v}(0)}}
  &=\frac12(\sigma_k+\gamma_k)(\sigma_p\sigma_l-\gamma_p\gamma_l)+\frac12(\sigma_k-\gamma_k)(\sigma_p\sigma_l+\gamma_p\gamma_l
    -2\sigma_p\gamma_l-2\gamma_p\sigma_l).
\end{align}
We easily see that the self-energy is then divergent in the ultraviolet
regime.
Indeed,  for large $t$,
$h_\kb(t,s)$ is asymptotic to a nonzero constant, $e_{\frac{t+s}{2}}$ and
  $e_{\frac{t-s}{2}}$ behave as $t^2$ and we have $t^2$ in the
    numerator. Therefore,  \eqref{eq:Sigma_t_l.} is linearly divergent
    at large $t$. We should not be surprised---contact potentials are
    not true potentials, they need a renormalization of the coupling
    constant, therefore they may lead to problems.

    For generic potentials   $h_k$ is given by \eqref{eqn:h(p)}:
      \begin{align}
    \label{eqn:h(p).}
h_\kb(\pb)&=  \sqrt{\frac{\mu \hat{v}^2(\kb) 
            }{\hat{v}(0)}}\big(\gamma_\kb-\sigma_\kb\big) 
            \big(\gamma_\pb\sigma_{\kb-\pb}+\sigma_\pb \gamma_{\kb-\pb}\big)\\
    &+ \sqrt{\frac{\mu \hat{v}^2(\pb) 
      }{\hat{v}(0)}}\big(\sigma_\kb\sigma_\pb\sigma_{\kb-\pb}-\gamma_\kb\gamma_\pb\gamma_{\kb-\pb}+\gamma_\kb\sigma_\pb\gamma_{\kb-\pb}
      -\sigma_\kb\gamma_\pb\sigma_{\kb-\pb}\big)\nn\\
          &+\sqrt{\frac{\mu \hat{v}^2(\kb-\pb) 
      }{\hat{v}(0)}}\big(\sigma_\kb\sigma_\pb\sigma_{\kb-\pb}-\gamma_\kb\gamma_\pb\gamma_{\kb-\pb}+\gamma_\kb\gamma_\pb\sigma_{\kb-\pb}
      -\sigma_\kb\sigma_\pb\gamma_{\kb-\pb}\big).\nn
  \end{align}
If we assume \eqref{assu2}, then
 $\hat v$ decays sufficiently fast and provides  a natural cutoff, so
 that
 the self-energy is
well-defined. We formulate this as a theorem:

\begin{theorem} \label{thm:finite}
  Suppose that Assumption  (\ref{assum}) holds.
  Then for $\kb\neq0$, the  self-energy  $\Sigma_{\kb}(z)$
 for
  $\mathrm{Im}z>0$ is finite.
One can also take its limit 
on the real line:
  \begin{equation}
   \Sigma_{\kb}(e_\kb+\ri0):=\lim_{\epsilon\searrow0}\Sigma_{\kb}(e_\kb+\ri\epsilon)
    .\label{seemoa}
  \end{equation}

  The same is true  the cutoff self-energy
  $\Sigma_{\kb}^\Lambda(z)$ involving contact potentials.
 \end{theorem}

\proof Let us sketch a proof of the first statement. For large $|\kb|$
we have
\be\sigma_k\simeq1,\qquad\gamma_k\simeq\frac{2\hat
  v(\kb)}{k^2}.\label{bebe}\ee
$h_k(t,s)^2$ contains several terms. Those containing
$\gamma_{\frac{t\pm s}{2}}$ are integrable because of \eqref{bebe}
and \eqref{assu2}. 
The only dangerous terms in $h_k(t,s)^2$  are
\be\frac{\mu\hat v(\frac{t\pm s}{2})^2}{
  \hat
  v(0)}\sigma_k^2\sigma_{\frac{t+s}2}^2\sigma_{\frac{t-s}2}^2.\ee
They are integrable by \eqref{assu2}. \qed
 
Unfortunately, there are also bad news. The energy shift has a
non-physical feature: it diverges as $\kb\to0$, as follows from the
theorem below. Therefore, we cannot treat seriously the results
obtained from the Fermi Golden Rule  concerning the real part of the
excitation spectrum.

\begin{theorem}\label{thm:renormalizationinfinity}
   Suppose that Assumption  (\ref{assum}) holds.
   Then  \begin{equation}
    \lim\limits_{k\to0}\Sigma_{\kb}(0)=-\infty.\label{seemo1}
  \end{equation}
\end{theorem}

   \begin{remark}
    The same statement could be obtained for more general potentials,
    e.g. satisfying  Assumption
  \eqref{assu-generic} instead of \eqref{assu-flat}
\end{remark}

  First note that under Assumption \eqref{assu-flat},
  for $|\kb|,|\pb|,|\kb-\pb|<\Lambda$  we have
  \be e_k=\eqref{strictly},\qquad h_k(\pb)=\eqref{below},\ee
  as for contact potentials.

Let $|\kb|<\Lambda/2$ and let us split the integral for the self-energy
as \be \eqref{eq:Sigma_t_l.}=\int_0^{\Lambda/2}\d
t+\int_{\Lambda/2}^\infty\d t.
\label{eq:Sigma_t_l..}
\ee
Taking into account $|s|<|\kb|$ we see
that the first integral involves only  $|\kb|,|\pb|,|\kb-\pb|<\Lambda$.
Therefore, in this integral all quantities such as $e_k$, $\sigma_k$,
$\gamma_k$,   $e_p$, $\sigma_p$,
$\gamma_p$,  $e_l$, $\sigma_l$,
$\gamma_l$,    
are as for contact
potentials.
Let us prove some lemmas about these quantities.

\begin{lemma}\label{sss} For small $p,l$, we have
  \begin{align}
        \frac{e_{\frac{t}{2}}}{e_p+e_l}-\frac12&=O(s^2),\label{iu1}\\
    \frac{pl}{e_pe_l}-\frac{t^2}{4e_{\frac{t}{2}}^2}&=O(s^2),\label{iu2}\\
    \sigma_p\sigma_l\sqrt{e_pe_l}-\sigma_{\frac{t}{2}}^2e_{\frac{t}{2}}&=O(s^2),\label{iu3}\\
    \gamma_p\gamma_l\sqrt{e_pe_l}-\gamma_{\frac{t}{2}}^2e_{\frac{t}{2}}&=O(s^2).          \label{iu4}                                   \end{align}
\end{lemma}

\proof We can assume that $s\geq0$.
\begin{equation}
  e_p'=\big(\tfrac{p^2}{2}+\mu\big)\big(\tfrac{p^2}{4}+\mu\big)^{-\frac12},\quad
  e_p''=p\big(\tfrac{p^2}{8}+\tfrac{3\mu}{4}\big)\big(\tfrac{p^2}{4}+\mu\big)^{-\frac32}=O(p).  
\end{equation}
Therefore,
\begin{equation*}
  2e_{\frac{t}{2}}-e_p-e_l =-\int_{-\frac{s}{2}}^{\frac{s}{2}}\big(\tfrac{s}{2}-|v|\big)e_{\frac{t}{2}+v}''\d v = O(ts^2),
  \end{equation*}
 and hence
 \begin{equation*}
    \frac{e_{\frac{t}{2}}}{e_p+e_l}-\frac12 =\frac{  2e_{\frac{t}{2}}-e_p-e_l }{2(e_p+e_l)}
\end{equation*}
is $O(s^2)$, which proves \eqref{iu1}.

Next, set $f(p):=\frac{p}{e_p}$. We have
\begin{equation}
  \frac{\d}{\d p}f(p)=\frac{-2p}{(p^2+4\mu)^{\frac32}}=O(p),\qquad
  \frac{\d^2}{\d p^2}f(p)=\frac{4(p^2-2\mu)}{(p^2+4\mu)^{\frac52}}=O(1).
\end{equation}
Hence
\begin{align}\label{iu5}
  &\frac{pl}{e_pe_l}-\frac{t^2}{4e_{\frac{t}{2}}^2}=
                                                    f(p)f(l)-f\big(\tfrac{t}{2}\big)^2\\
  =&\int_0^{\frac{s}{2}}\big(\tfrac{s}{2}-v\big)\Big(
  f''\big(\tfrac{t}{2}+v\big)   f\big(\tfrac{t}{2}-v\big)-
  2  f'\big(\tfrac{t}{2}+v\big)   f'\big(\tfrac{t}{2}-v\big)
  +f\big(\tfrac{t}{2}+v\big)   f''\big(\tfrac{t}{2}-v\big)\Big)\d v,&
                             \nn                   \end{align}
which is $O(s^2)$, which proves \eqref{iu2}.

We check that the 0th, 1st and 2nd derivatives of
\begin{align}
  \sigma_p\sqrt{e_p}=\frac{1}{\sqrt2}\sqrt{\tfrac{p^2}{2}+\mu+\sqrt{\tfrac{p^4}{4}+\mu 
  p^2}},\\
    \gamma_p\sqrt{e_p}=\frac{1}{\sqrt2}\sqrt{\tfrac{p^2}{2}+\mu-\sqrt{\tfrac{p^4}{4}+\mu 
  p^2}}
  \end{align}
are bounded. Then we argue as in 
\eqref{iu5}, proving \eqref{iu3} and \eqref{iu4}. \qed

\begin{lemma}\label{quw30}
\begin{align}
\lim_{k\to0}  \int_k^\Lambda\d 
                                          t\int_{-k}^k\d
    s\frac{(\sigma_p\sigma_l-\gamma_p\gamma_l)^2
    pl                  }{8k(e_p+e_l)}=\int_0^\Lambda\d
  t\frac{t^2}{64e_{\frac{t}{2}}},
\label{quw3}  \end{align}
where the right hand side is a finite positive number. \end{lemma}

\proof
We have
\begin{align}
  &\frac{(\sigma_p\sigma_l-\gamma_p\gamma_l)^2 
  pl                  }{8k(e_p+e_l)}- \frac{ t^2}{8\cdot 8ke_{\frac{t}{2}}}\\
  =&
\frac{\big((\sigma_p\sigma_l-\gamma_p\gamma_l)\sqrt{e_pe_l}+ e_{\frac{t}{2}}\big)
     pl    }{8k(e_p+e_l)e_pe_l}
       \Big((\sigma_p\sigma_l-\gamma_p\gamma_l)\sqrt{e_pe_l}- e_{\frac{t}{2}}\big)
                  \Big)\label{tet1}\\
  +&\frac{ e_{\frac{t}{2}}^2}{8k(e_p+e_l)}\Big(\frac{pl}{e_pe_l}-\frac{t^2}{4e_{\frac{t}{2}}^2}\Big)\label{tet2}\\
  +&\frac{ t^2}{32ke_{\frac{t}{2}}}\Big(\frac{e_{\frac{t}{2}}}{e_p+e_l}-\frac12\Big).\label{tet3}
\end{align}
By Lemma \ref{sss} the terms in the big brackets on the right of \eqref{tet1},
\eqref{tet2} and \eqref{tet3} are $O(s^2)$. The terms in
 \eqref{tet2},
\eqref{tet3}  on the left are all $\frac{1}{k}O(t)$. The most singular in $t$ term is the one on the left of  \eqref{tet1} and it is of order $\frac{1}{k}O(t^{-1})$.  Therefore,
\begin{align}
&  \int_k^\Lambda\d 
                                          t\int_{-k}^k\d
    s\Bigg(\frac{(\sigma_p\sigma_l-\gamma_p\gamma_l)^2
    pl
  }{8k(e_p+e_l)}-\frac{ t^2}{64e_{\frac{t}{2}}}\Bigg)\\
  =&  \int_k^\Lambda\d 
                                          t\int_{-k}^k\d 
     s  O(t^{-1})\frac{O(s^2)}{k}
     = \int_k^\Lambda\d 
                                          t O(t^{-1}k^2)=O(k^2\ln k)\to0.
\label{quw3a}  \end{align}

\qed

 \noindent{\it Proof of Theorem \ref{thm:renormalizationinfinity}.}
  
   The second integral on the right of \eqref{eq:Sigma_t_l..} is
   convergent as $k\to0$. Let us consider the first integral:
 \begin{align} 
& \frac{1}{ 2 (2\pi)^2}\int_k^{\Lambda/2}\d
                                          t\int_{-k}^k\d s\frac{h_k(t,s)^2
                      (t^2-s^2)                   }{8k(z-e_{\frac{t+s}{2}}-e_{\frac{t-s}{2}})}\\
  =&(\sigma_k+\gamma_k)^2 \int_k^{\Lambda/2}\d 
                                          t\int_{-k}^k\d
    s\frac{(\sigma_p\sigma_l-\gamma_p\gamma_l)^2
    pl                  }{2k(e_p+e_l)}\label{quw1}\\
  +&2  \int_k^{\Lambda/2}\d 
                                          t\int_{-k}^k\d
    s\frac{(\sigma_p\sigma_l-\gamma_p\gamma_l) (\sigma_p\sigma_l+\gamma_p\gamma_l-2\sigma_p\gamma_l-2\gamma_p\sigma_l)
    pl                  }{2k(e_p+e_l)}\label{quw2}\\+&
  (\sigma_k-\gamma_k)^2
\int_k^{\Lambda/2}\d 
                                          t\int_{-k}^k\d
     s\frac{(\sigma_p\sigma_l+\gamma_p\gamma_l-2\sigma_p\gamma_l-2\gamma_p\sigma_l)^2
     pl                  }{2k(e_p+e_l)}\label{quw}
\end{align}
where we used that $\sigma_k^2 - \gamma_k^2 =1$. Since ${\Lambda/2}$ is fixed we are only interested in the small $t$ region. Since $k$ is small too, this implies also $p$ and $l$ are small. For such we have 
\begin{align}\label{quww1} (\sigma_k+\gamma_k)^2&\geq Ck^{-1},\quad C>0\\\label{quww2}
  (\sigma_k-\gamma_k)^2&=O(k),\\\label{quww4}
  (\sigma_p\sigma_l-\gamma_p\gamma_l)\sqrt{pl}&=O(p)+O(l)=O(t),\\\label{quww5}
  (\sigma_p\sigma_l+\gamma_p\gamma_l-2\sigma_p\gamma_l-2\gamma_p\sigma_l)
  \sqrt{ pl              }&=O(1), \\
  \frac{1}{e_p+e_l}&=O(t^{-1}). \label{quww6}\end{align}
By Lemma \ref{quw30} and \eqref{quww1},
                            \begin{align}
                              |\eqref{quw1}|&\geq C_1 k^{-1}\to
                                              +\infty.
                            \end{align}
                            By  \eqref{quww4},
                            \eqref{quww5} and \eqref{quww6},
\begin{align}
                              |\eqref{quw2}|&\leq C 
\int_k^{\Lambda/2}\d 
                                          t\int_{-k}^k\d
                                              s\frac{1}{k} \rightarrow C_\Lambda \qquad \text{as}\,\,\, k\to 0.
 \end{align}

                                                                                        By \eqref{quww2},
                            \eqref{quww5} and \eqref{quww6},
\begin{align}|\eqref{quw}|&\leq C k
\int_k^\Lambda\d 
                                          t\int_{-k}^k\d
                                              s\frac{1}{kt}\leq Ck|\ln(k)|\to0,
                            \end{align}
Hence \eqref{eq:Sigma_t_l.} converges to $-\infty$.
                            \qed
\\
\\
\\
\noindent
{\bf Acknowledgements.} The work of all authors was supported by the Polish-German NCN-DFG grant Beethoven Classic 3 (project
no. 2018/31/G/ST1/01166).  We thank the anonymous referees for many helpful remarks.
\\
\\
\noindent
{\bf Data Availability} This manuscript has no associated data.\\
\\
\noindent
{\bf Declarations}

{\bf Conflict of interest} This manuscript has no Conflict of interest.

\bibliographystyle{siam}

\begin{thebibliography}{10}

 
\bibitem{BauCom-17}
{\sc H.~H.~Bauschke and P.~L.~Combettes}, {\em Convex Analysis and Monotone Operator Theory in Hilbert Spaces}, CMS Books in Mathematics, Springer, 2017
 
  
  
\bibitem{Beliaev-58}
{\sc S.T. Beliaev}, {\em Energy spectrum of a non-ideal Bose gas}, Sov. Phys. JETP 34 (7), 299 (1958)

 




 


 


\bibitem{BocBreCenSch-19}
{\sc C.~{Boccato}, C.~Brennecke,  S.~{Cenatiempo} and B.~{Schlein}}, {\em  Bogoliubov theory in the Gross--Pitaevskii limit}, Acta Mathematica 222 (2), 219-335 (2019)


 


\bibitem{BocBreCenSch-20a}
{\sc C.~{Boccato}, C.~Brennecke,  S.~{Cenatiempo} and B.~{Schlein}}, {\em The excitation spectrum of Bose gases interacting through singular potentials}, J. Eur. Math. Soc. 22(7), 2331-2403 (2020)


 


\bibitem{Bogoliubov-47}
{\sc N.~N. Bogoliubov}, {\em On the theory of superfluidity}, J. Phys. (USSR),
  11 (1947), p.~23.

 
\bibitem{BosLeoPetRad-23}
 {\sc L. Bossmann, N. leopold,  S. Petrat,  and S. Rademacher}, {\em Ground state of Bose gases interacting through singular potentials}, arXiv:2309.12233 [math-ph]


 
\bibitem{BosPetPicSof-22}
 {\sc L. Bossmann, S. Petrat, P. Pickl and A. Soffer}, {\em Beyond Bogoliubov Dynamics}. Pure and Applied Analysis 3 (4), 677-726 (2022)


 

\bibitem{BosPetSei-21} {\sc L. Bossmann, S. Petrat and R. Seiringer}, {\em Asymptotic expansion of the low-energy excitation spectrum for weakly interacting bosons}. Forum of Mathematics, Sigma 9, e28 (2021)

\bibitem{BreCapSch-22}
{\sc  C.~Brennecke, M.~Caporaletti and B.~{Schlein}}, {\em Excitation spectrum of Bose gases beyond the Gross–Pitaevskii regime}, Rev. Math. Phys. 34 (9), 2250027 (2022)

\bibitem{BreSchSch-22}
{\sc  C.~Brennecke, B.~{Schlein} and S. Schraven}, {\em Bogoliubov Theory for Trapped Bosons in the Gross--Pitaevskii Regime}, Ann. Henri Poincare 23, 1583 (2022)



\bibitem{CarOlgSASch-23}
{\sc  C.~Caraci, A.~Olgiati, D.~Saint-Aubin and B.~{Schlein}}, {\em Third order corrections to the ground state energy of a Bose gas in the Gross--Pitaevskii regime}, arXiv preprint arXiv:2311.07433



\bibitem{CorDerZin-09}
{\sc H.~D.~Cornean,  J.~Derezi\'nski and  P.~Zi\'n}, {\em On the infimum of the energy-momentum spectrum of a homogeneous Bose gas}, J. Math. Phys.  50, 062103 (2009) 

\bibitem{Davies-74}
{\sc E.~B.~Davies}, {\em Markovian master equations},  Commun. Math. Phys. 39, 91–110 (1974)

 
\bibitem{Chung-09}
{\sc M.-C.  Chung and A. B.  Bhattacherjee}, {\em Damping in 2d and 3d dilute Bose gases}, New Journal of Physics 11 (2009), no. 12, 123012.

 \bibitem{DerdeR-07}
{\sc J.~{Derezi{\'n}ski} and W.~de~Roeck}, {\em Extended Weak Coupling Limit for Friedrichs Hamiltonians}, J . Math. Phys. 48 (2007) 012103

\bibitem{DerdeR-08}
{\sc J.~{Derezi{\'n}ski} and W.~de~Roeck}, {\em Extended Weak Coupling Limit for Pauli-Fierz Hamiltonians}, Comm. Math. Phys. 279 (2008) 1-30

 \bibitem{DerFru-02}
{\sc J.~{Derezi{\'n}ski} and R.~Fr\"{u}boes}, {\em Renormalization of Friedrichs Hamiltonians}, Rep. Math. Phys. 50 (2002) 433-438


\bibitem{DerMeiNap-13}
{\sc J.~Derezi\'nski, K.A. Meissner and   M.~{Napi{\'o}rkowski}}, {\em On the   energy-momentum spectrum of a homogeneous Fermi gas}, Ann. Henri Poincar\'e 14 (1), 1-36 (2013)


\bibitem{DerNap-13}
{\sc J.~{Derezi{\'n}ski} and M.~{Napi{\'o}rkowski}}, {\em {Excitation spectrum
  of interacting bosons in the mean-field infinite-volume limit}}, Ann.
  Henri Poincar\'e, 15 (2014), pp.~2409--2439.
\newblock Erratum: Ann. Henri Poincar\'e 16 (2015), pp. 1709-1711.


  \bibitem {DeGe} 
{\sc J.~{Derezi{\'n}ski} and C.~G\'erard}, {\em  Mathematics of Quantization and Quantum
      Fields}, Cambridge Monographs in Mathematical Physics, Cambridge
      University Press 2013


 \bibitem{dirac} 
{\sc P.A.M. Dirac}, {\em The Quantum Theory of Emission and Absorption of Radiation},  Proceedings of the Royal Society A. 114 (767): 243-265

\bibitem{FetWal-03}
{\sc A.L. Fetter and J.D. Walecka}, {\em Quantum Theory of Many-Particle Systems}, Dover Publications, 2003

\bibitem{fermi}
  {\sc E.  Fermi}, {\em Nuclear Physics}, University of Chicago Press,
  1950

\bibitem{Friedrichs-65}
{\sc K.~O.~Friedrichs},  {\em Perturbation of Spectra in Hilbert Space}, American Mathematical Society, Providence, (1965)

\bibitem{Giorgini-98}
{\sc S. Giorgini}, {\em Damping in dilute Bose gases: A mean-field approach}, Phys. Rev. A 57, 2949 (1998)

\bibitem{GreSei-13}
{\sc P.~Grech and R.~Seiringer}, {\em The excitation spectrum for weakly
  interacting bosons in a trap}, Comm. Math. Phys., 322 (2013), pp.~559--591.

\bibitem{Godfrinetall-21}
  {\sc H. Godfrin, K. Beauvois, A. Sultan,E. Krotscheck, J. Dawidowski  B. Fak, and J. Ollivier}
{\em Dispersion relation of Landau elementary excitations and
thermodynamic properties of superfluid 4He}, Phys. Rev. B 103, 104516 (2021)



\bibitem{Hodbyetall-01}
{\sc E.~Hodby, O.~M.~Marago, G.~Hechenblaikner and C.~J.~Foot}, {\em Experimental Observation of Beliaev Coupling in a Bose-Einstein Condensate}, Phys. Rev. Lett. 86, 2196 (2001)

\bibitem{HohMar-65}
{\sc P.~C. Hohenberg and P.~C. Martin}, {\em Microscopic theory of superfluid helium}, Ann. Phys. 34 (2), 291-359 (1965)



\bibitem{HugPin-59}
{\sc N.~M.~Hugenholtz and D.~Pines}, {\em Ground-State Energy and Excitation Spectrum of a System of Interacting Bosons}, Phys. Rev. 116, 489 (1959)

\bibitem{Katzetall-02}
{\sc N.~Katz, J.~Steinhauer, R.~Ozeri and N.~Davidson}, {\em Beliaev Damping of Quasiparticles in a Bose-Einstein Condensate}, Phys. Rev. Lett. 89, 220401 (2002)


\bibitem{LewNamSerSol-15}
{\sc M.~Lewin, P.~T. Nam, S.~Serfaty and J.~P. Solovej}, {\em Bogoliubov
  spectrum of interacting {B}ose gases}, Comm. Pure Appl. Math., 68 (2015),
  pp.~413--471.

\bibitem{LieSeiYng-05}
{\sc E.~H. Lieb, R.~Seiringer and J.~Yngvason}, {\em Justification of $c$-Number Substitutions in Bosonic Hamiltonians}, Phys. Rev. Lett. 94, 080401 (2005)

\bibitem{Liu-97}
{\sc W.V. Liu}, {\em Theoretical study of the damping of collective excitations in a Bose-Einstein Condensate} Phys. Rev. Lett. 79 (1997), 4056-4059


\bibitem{MohMor-60}
{\sc   F. Mohling  and M. Morita}, {\em Temperature Dependence of the Low-Momentum Excitations in a Bose Gas of Hard Spheres}, Phys. Rev. 120 (3), 681-688 (1960)


\bibitem{NamNap-17b}  
{ \sc P.~T. Nam and M.~Napi\'orkowski}, {\em Norm approximation for many-body quantum dynamics and Bogoliubov theory}, in:  "Advances in Quantum Mechanics: contemporary trends and open problems", Springer. (2017)



\bibitem{NamNap-21}  
{ \sc P.~T. Nam and M.~Napi\'orkowski}, {\em Two-term expansion of the ground state one-body density matrix of a mean-field Bose gas},  Calc. Var. PDE  60 (3), 1-30 (2021)




\bibitem{NamSei-15}
{\sc P.T. Nam and R. Seiringer}, {\em Collective excitations of Bose gases in the mean-field regime}, Arch. Ration. Mech. Anal. 215 (2), 381-417 (2015)

\bibitem{NamTri-23}
{\sc P.~T. Nam and A. Triay}, {\em  Bogoliubov excitation spectrum of trapped Bose gases in the Gross-Pitaevskii regime}, J. Math. Pures Appl. 176,  18-101 (2023)


\bibitem{Napiorkowski-23}
{\sc M. Napi\'orkowski}, {\em Dynamics of interacting bosons: a compact review},  in: Density Functionals for Many-Particle Systems - Mathematical Theory and Physical Applications of Effective Equations, World Scientific (2023)

 \bibitem{reed-simon-4} 
{\sc M. Reed and B. Simon,}, {\em Methods of Modern   Mathematical Physics. IV: Analysis of Operators}, Academic Press, 1978

\bibitem{Seiringer-11}
{\sc R.~Seiringer}, {\em The excitation spectrum for weakly interacting
  bosons}, Commun. Math. Phys., 306 (2011), pp.~565--578.

\bibitem{ShiGri-98}
{\sc H. Shi and A. Griffin}, {\em Finite-temperature excitations in a dilute Bose-condensed gas}, Phys. Rep. 304 (1998) 1-87
  
\bibitem{wigner-weisskopf}
  {\sc V. Weisskopf, E.P.Wigner,} {\em Berechnung der nat\"urlichen
  Linienbreite auf Grund der Diracschen Lichtteorie}, Z. Phys. 63
  (1930) 54-73




\end{thebibliography}

\end{document}